\documentclass{article}

\usepackage{amsfonts}
\usepackage{amsmath,amssymb,bbm}
\usepackage{mathtools}
\usepackage[utf8]{inputenc}
\usepackage[T1]{fontenc}
\usepackage[english]{babel}
\usepackage{dsfont}
\usepackage{url}            
\usepackage{booktabs}       
\usepackage{amsfonts}       
\usepackage{nicefrac}       
\usepackage{microtype}      
\usepackage{amsmath, amssymb}
\usepackage{fullpage}
\usepackage{authblk}
\usepackage{amsfonts}
\usepackage{amsthm}
\usepackage{bbm}
\usepackage{dsfont}
\usepackage{graphicx}
\usepackage{float}
\usepackage{xcolor}
\usepackage{subcaption}
\usepackage{wrapfig}
\usepackage{tikz}
\usepackage{pst-node}
\usepackage{algorithm}
\usepackage{listings}
\usepackage{fancyvrb}
\fvset{fontsize=\normalsize}
\usepackage[noend]{algpseudocode}
\usetikzlibrary{fit,positioning,bayesnet}
\usetikzlibrary{automata,arrows}
\usetikzlibrary{backgrounds}

\usepackage[acronym,nowarn]{glossaries}
\usepackage{multirow}
\usepackage{multicol}
\usepackage{bm}
\usepackage{upgreek}
\usepackage{natbib}
\usepackage[pagebackref=true]{hyperref}
\fvset{fontsize=\normalsize}
\hypersetup{colorlinks,linkcolor={blue},citecolor={red},urlcolor={blue}}
\setcitestyle{numbers,comma,square}
\usepackage{cleveref}

\usepackage{enumitem}
\setlist[enumerate]{leftmargin=.5in}
\setlist[itemize]{leftmargin=.5in}

\usepackage{amsopn}

\DeclareMathOperator*{\argmax}{\arg\!\max}
\DeclarePairedDelimiterX{\DistanceRelation}[2]{(}{)}{#1\;\delimsize\|\;#2}

\newacronym{ULA}{ULA}{unadjusted Langevin algorithm}
\newacronym{pULA}{pULA}{preconditioned unadjusted Langevin algorithm}
\newacronym{MALA}{MALA}{Metropolis adjusted Langevin algorithm}
\newacronym{IPLA}{IPLA}{interacting particle Langevin algorithm}
\newacronym{PGD}{PGD}{particle gradient descent}
\newacronym{statFEM}{statFEM}{statistical finite element method}
\newacronym{EM}{EM}{expectation-maximisation}
\newacronym{MMLE}{MMLE}{marginal maximum likelihood estimation}
\newacronym{MCMC}{MCMC}{Markov chain Monte Carlo}
\newacronym{SOUL}{SOUL}{stochastic optimization via unadjusted Langevin}

\newacronym{FEM}{FEM}{finite element method}
\newacronym{ODE}{ODE}{ordinary differential equation}
\newacronym{SDE}{SDE}{stochastic differential equation}
\newacronym{PDE}{PDE}{partial differential equation}
\newacronym{QoI}{QoI}{quantity-of-interest}
\newacronym{MAP}{MAP}{\textit{maximum a posteriori}}
\newacronym{MMAP}{MMAP}{marginal \textit{maximum a posteriori}}
\newacronym{GP}{GP}{Gaussian process}
\newacronym{UQ}{UQ}{uncertainty quantification}
\newacronym{MC}{MC}{Monte Carlo}

\newcommand{\vu}{\mathbf{u}}
\newcommand{\vx}{\mathbf{x}}
\newcommand{\vy}{\mathbf{y}}
\newcommand{\vz}{\mathbf{z}}
\newcommand{\vb}{\mathbf{b}}
\newcommand{\ve}{\mathbf{e}}
\newcommand{\vv}{\mathbf{v}}
\newcommand{\vr}{\mathbf{r}}
\newcommand{\vm}{\mathbf{m}}
\newcommand{\vs}{\mathbf{s}}
\newcommand{\vp}{\mathbf{p}}

\newcommand{\mA}{\mathbf{A}}
\newcommand{\mC}{\mathbf{C}}
\newcommand{\mG}{\mathbf{G}}
\newcommand{\mH}{\mathbf{H}}
\newcommand{\mM}{\mathbf{M}}
\newcommand{\mP}{\mathbf{P}}
\newcommand{\mR}{\mathbf{R}}
\newcommand{\mF}{\mathbf{F}}
\newcommand{\mJ}{\mathbf{J}}
\newcommand{\mSigma}{\boldsymbol{\Sigma}}

\newcommand{\md}{\mathrm{d}}

\newcommand{\cF}{\mathcal{F}}

\newcommand{\bE}{\mathbb{E}}
\newcommand{\bR}{\mathbb{R}}

\newtheorem{theorem}{Theorem}
\newtheorem{remark}{Remark}

\newtheorem{lemma}{Lemma}

\title{Statistical Finite Elements via Interacting Particle Langevin Dynamics}

\author{Alex Glyn-Davies, Connor Duffin, Ieva Kazlauskaite, Mark Girolami, \"{O}. Deniz Akyildiz}

\begin{document}
\maketitle

\begin{abstract}
In this paper, we develop a class of interacting particle Langevin algorithms to solve inverse problems for partial differential equations (PDEs). In particular, we leverage the statistical finite elements (statFEM) formulation to obtain a finite-dimensional latent variable statistical model where the parameter is that of the (discretised) forward map and the latent variable is the statFEM solution of the PDE which is assumed to be partially observed. We then adapt a recently proposed expectation-maximisation like scheme, interacting particle Langevin algorithm (IPLA), for this problem and obtain a joint estimation procedure for the parameters and the latent variables. We consider three main examples: (i) estimating the forcing for linear Poisson PDE, (ii) estimating the forcing for nonlinear Poisson PDE, and (iii) estimating diffusivity for linear Poisson PDE. We provide computational complexity estimates for forcing estimation in the linear case. We also provide comprehensive numerical experiments and preconditioning strategies that significantly improve the performance, showing that the proposed class of methods can be the choice for parameter inference in PDE models.
\end{abstract}

\section{Introduction}
\label{sec:intro}
Statistical estimation methods for \gls*{PDE} models are of significant recent interest as such approaches offer powerful tools to analyse complex systems which can be jointly described through both mechanistic and data-driven components. They allow for the estimation of partially or indirectly observed quantities-of-interest, which, depending on the setting may be a model input parameter (\textit{inversion})~\cite{stuart2010inverse,tarantola2005inverse}, or, a model state (\textit{data assimilation})~\cite{sanz-alonso2023Inversea}. In this work we are interested in solving these problems jointly, i.e., the joint parameter and state estimation problem within the context of static stochastic \glspl*{PDE}. We aim at tackling the problem of inference for \glspl*{PDE} using a recent statistical construction termed the \gls*{statFEM}~\cite{akyildiz2022statistical,duffin2021statistical,girolami2021statistical}, which provides a statistical model of the \gls*{PDE} solution. Using this statistical model, we aim at solving the joint state and parameter estimation problem using the \gls*{IPLA}~\cite{akyildiz2025Interacting}, a recently proposed class of statistical estimation methods for latent variable models.

To set up the context, consider the following stochastic \gls*{PDE} with the differential operator $\mathcal{L}_\theta$ parameterised by $\theta$, and forcing $f$
\begin{align}\label{eq:pde-model}
  \mathcal{L}_\theta u = f + \epsilon, \quad
  \epsilon \sim \mathcal{GP}(0, k)
\end{align}
where $k$ is a (regular enough) kernel function (see section~\ref{subsec:statfem-linear} for details). A data generating process can then be defined through the continuous \textit{observation operator} $\mathcal{H} \colon L^2(\Omega) \rightarrow \mathbb{R}^{n_y}$ 
\begin{align}\label{eq:observation-model} 
  \vy = \mathcal{H}(u) + \vr, \quad
  \vr \sim \mathcal{N}(\mathbf{0},\mathbf{R}).
\end{align}
The equations \eqref{eq:pde-model}--\eqref{eq:observation-model} specify a statistical model, which can be used to infer the latent PDE solution $u$. This is a standard problem, having been the subject of much recent attention with, for example, the \gls*{statFEM}~\cite{akyildiz2022statistical,duffin2021statistical,girolami2021statistical}. While these methods can be generically used for sampling the parameters jointly, there is a gap in the literature for dedicated methods to obtain point estimates of the parameters directly, which may avoid a separate procedure to get maximisers from samples. Denote by $\vz$ these unknown parameters, which may comprise any unknown model components, e.g. PDE parameters $\theta$, or the forcing function $f$.
Discretising~\eqref{eq:pde-model} using \gls*{statFEM}, we obtain a model state vector $\vu$, and conditional distribution $p(\vu|\vz)$. To estimate the parameters of the \gls*{PDE} (that are denoted as $\vz$), we aim at maximising the marginal likelihood by solving the following optimisation problem
\begin{equation}
  \label{eq:mmle}
  \vz_{\mathrm{MMLE}} = \argmax_{\vz \in \mathcal{Z}} \log p(\vy | \vz), \quad \text{where} \quad p(\vy | \vz) := \int p(\vy | \vu) \, p(\vu | \vz) \md \vu,
\end{equation}
which is termed the \gls*{MMLE} problem. When we include a prior distribution in the model, the \gls*{MMAP} problem can be solved by maximising the marginal posterior distribution
\begin{equation}
    \label{eq:mmap}
  \vz_{\mathrm{MMAP}} = \argmax_{\vz \in \mathcal{Z}} \log p(\vz , \vy), \quad \text{where} \quad p(\vz , \vy) := \int p(\vy | \vu) \, p(\vu | \vz) p(\vz) \md \vu.
\end{equation}
In this paper we will focus on this estimation, using the so-called \gls*{IPLA}, which also provides us with estimates of the posterior $p(\vu | \vy, \vz)$. Our choice of \gls*{IPLA} is motivated by its theoretical guarantees, which allow us to conduct some analysis in our proposed model setting. To set the scene for this work, we now go through a review of the relevant literature.

\subsection{Statistical Finite Elements}
Introduced in \cite{girolami2021statistical}, the \gls*{statFEM} was developed for calibrating finite element models using observational data. Uncertainties from model misspecification and observational noise are assumed, and are used to construct a statistical model describing how the stochastic \gls*{FEM} solution generates the data. A Bayesian perspective is taken, and the prior over the \gls*{FEM} solution is updated based on observations via posterior inference. This has been extended to time-varying \glspl*{PDE} in \cite{duffin2021statistical}, where filtering techniques are used to sequentially update the posterior distribution. The Bayesian framework also allows for estimation of random field parameters based on the marginal likelihood of the data \cite{duffin2021statistical, febrianto2022Digital, girolami2021statistical}. Recently, theoretical advances have been made, including error analysis in \cite{karvonen2022error} and convergence analysis with mesh refinement in \cite{papandreou2023Theoretical}.
Langevin dynamics has been proposed as a method for efficient sampling from the posterior distribution \cite{akyildiz2022statistical}, which includes analysis of convergence of the sampler to the target distribution.

\subsection{Marginal MLE and MAP}
In this section, we first review the statistical literature pertaining to the \gls*{MMLE}. The estimation problem in~\eqref{eq:mmle} can be solved with the classical \gls*{EM} algorithm~\cite{dempster1977maximum}, or, the stochastic/Monte Carlo variants thereof~\cite{celeux1992stochastic,wei1990monte}. For a significant number of models, the \gls*{EM} can be implemented via stochastic approximations when the posterior distribution is tractable and easy to sample from. When this sampling is infeasible, the E-step can be approximated using \gls*{MCMC} methods~\cite{delyon1999convergence,atchade2017perturbed}, which provide a rich framework for bespoke inference. Recent extensions have considered \textit{unadjusted} \gls*{MCMC}, whereby the Metropolis adjustment is skipped for computational simplicity. Perhaps the most notable of these is the \gls*{SOUL} algorithm~\cite{de2021efficient}, which makes use of the \gls*{ULA}~\cite{dalalyan2017theoretical,durmus2017nonasymptotic}, to solve \gls*{MMLE} problems. Through Fisher's identity the \gls*{SOUL} approach leverages a stochastic approximation, iteratively running the E-step with a \gls*{ULA} chain, then using these samples to compute the M-step. The main bottleneck of this method is that it requires a unique Markov chain to be run to compute the E-step, for each M-step.

A similar approach is taken in~\cite{kuntz2023Particle}, which replaces the coordinate-wise procedure of~\cite{de2021efficient} with an interacting particle system. This method constructs $N$ particles for the latent variables (instead of running a chain in time) and builds a procedure in the joint space of $N$ latent variables and parameters of interest. The authors demonstrate favourable convergence properties and computational advantages over \gls*{SOUL}, with additional theoretical results \cite{caprio2024error}. In this work, we follow a similar approach, closer to the standard Langevin dynamics, namely the \gls*{IPLA}~\cite{akyildiz2025Interacting}. This modifies the particle gradient descent of~\cite{kuntz2023Particle} through modifying the parameter update by noising it, hence forming a \gls*{SDE} instead of a system of mixed \glspl*{SDE} and \gls*{ODE} as in \cite{kuntz2023Particle}. Taking such an approach enables straightforward nonasymptotic results for estimating $\vz$, makes the analysis akin to the standard analysis of Langevin diffusions. This similarity allows us to adapt results from standard discretisation error analysis of the Langevin diffusion, which we leverage in this work.

The problem of joint state and parameter estimation has received attention from various communities, including physics~\cite{abarbanel2009Dynamical}, chemistry~\cite{dochain2003State}, electrical engineering~\cite{wan1997Dual,wan2000unscented}, data assimilation~\cite{moradkhani2005Dual,evensen2009ensemble}, and statistics~\cite{kantas2015Particle}. 
For the extension to \gls*{PDE} systems, Bayesian inverse problems for \glspl*{PDE} tackle the problem of parameter estimation (i.e. $\vz = \theta$) from data or \textit{inversion}, but typically assume a \textit{known} (and deterministic) \gls*{PDE} forward model, which is restrictive in the modelling of uncertainties. 
Estimation of \gls*{PDE} forcing/source terms from partial observations (i.e. $\vz = f$) is relevant to climate modelling, for example the estimation of greenhouse gas sources from satellite observations~\cite{nassar2011inverse} or identification of pollutant sources~\cite{atmadja2001pollution}. This is also of interest to the engineering community, for example in structural health monitoring for estimating unknown forcing from sensor measurements~\cite{erdogan2014investigation}.
By constructing a probabilistic forward model based on our assumed \gls*{PDE}, we allow for principled accounting for model misspecification, and joint estimation of parameters and uncertainty quantification for our \gls*{PDE} solution.

\subsection{Our contribution}
Previous work has shown that sampling from~\eqref{eq:pde-model}, and its resultant posterior distribution can be efficiently done using \gls*{ULA}~\cite{akyildiz2022statistical}. However, given observations and unknown parameters, the resulting \gls*{PDE} model is a latent variable statistical model where the inference of parameters and the latent field requires more than just sampling algorithms. We build on the interacting particle solutions which are shown to be efficient for \gls*{MMLE} problems~\cite{akyildiz2025Interacting,kuntz2023Particle,encinar2024proximal,oliva2024kinetic}.

Through the use of \gls*{IPLA} in combination with \gls*{statFEM}, this paper provides a novel algorithmic framework to conduct PDE state and parameter estimation, studied in the context of elliptic \glspl*{PDE}. Convergence of the schemes is analysed, and we provide thorough numerical evidence of their performance. More specifically:
\begin{itemize}
    \item We develop an algorithm for forcing estimation (section~\ref{subsec:linear-poisson}) and solving the inverse problem (section~\ref{subsec:inverse-problem}) in section~\ref{sec:ipla-linear-poisson}, which is applied to the linear Poisson \gls*{PDE}. 
    Our choice of Poisson \glspl*{PDE} is for demonstration since our algorithms can be similarly adapted to other \glspl*{PDE}. 
    \item We develop preconditioning strategies within this section to improve the condition number of the problem. We also introduce warm-start strategies to avoid the numerical difficulties caused by the ill-conditioning of the problem.
    \item We then provide, in section~\ref{subsec:convergence_forcing_MMAP}, the rate of convergence for the forcing estimation procedure presented in section~\ref{subsec:linear-poisson}. In particular, in Theorem~\ref{thm:convergence_forcing}, we prove a nonasymptotic bound for the \gls*{IPLA} procedure as applied to the forcing estimation problem and provide complexity estimates in Remark~\ref{rem:complexity}. Our analysis shows that, provided that the number of particles $N$ is chosen such that $N = \mathcal{O}(\varepsilon^{-2} n_u)$ and the step-size $\gamma$ of our method is chosen so that $\gamma = \mathcal{O}(\varepsilon^2 n_u^{-1} \kappa^{-2})$ where $n_u$ is the degrees of freedom and $\kappa$ is the condition number of the problem, our method provably achieves $\varepsilon$ error in $\widetilde{\mathcal{O}}(n_u \kappa^2 \varepsilon^{-2})$ steps. This also shows why preconditioning is crucial for our case. 
    \item We provide experiments to demonstrate the order of convergence proved in Theorem~\ref{thm:convergence_forcing} using an analytical example, as well as the benefit preconditioning through computation of the condition number (section~\ref{subsec:forcing-linear-results}). We show the benefit of using a warm-start for solving the inverse problem in section~\ref{subsec:diffusivity-estimation}.
    \item We develop an algorithm for forcing estimation in the nonlinear \gls*{statFEM} setting (section~\ref{subsec:statfem-nonlinear}), which we apply to the nonlinear Poisson \gls*{PDE} (section~\ref{subsec:nonlinear-poisson}). We compare three different nonlinear approximations for forcing estimation in section~\ref{subsec:forcing-nonlinear-results}.
\end{itemize}

This paper is organised as follows: in section \ref{sec:background} we introduce the \acrlong*{statFEM} for linear \glspl*{PDE} (section~\ref{subsec:statfem-linear}) before introducing Langevin dynamics for \gls*{statFEM} (section~\ref{subsec:technical-ULA}) and the \acrlong*{IPLA} (section~\ref{subsec:technical-IPLA}). 
Section~\ref{sec:ipla-linear-poisson} shows how the \gls*{statFEM} construction can be used to embed the probabilistic model for the linear Poisson PDE within the \gls*{IPLA} for \gls*{MMLE}/\gls*{MMAP} estimation. The algorithms and preconditioners are given in section~\ref{subsec:linear-poisson} for forcing estimation, and section~\ref{subsec:inverse-problem} for solving the inverse problem.
Theoretical results are derived in section~\ref{sec:theoretical-results}. Experimental results for the linear Poisson \gls*{PDE} are given in section~\ref{sec:experimental-results}, including the forcing estimation (section~\ref{subsec:forcing-linear-results}), and the inverse problem (section~\ref{subsec:diffusivity-estimation}) results. 
In section~\ref{sec:nonlinear}, we extend the methodology to nonlinear \gls*{statFEM}. Section~\ref{subsec:statfem-nonlinear} derives the form of the \gls*{statFEM} prior for the nonlinear Poisson \gls*{PDE}, and section~\ref{subsec:nonlinear-poisson} outlines three different approximations to this prior for use in the \gls*{IPLA} algorithm. Forcing estimation results are given in section~\ref{subsec:forcing-nonlinear-results}.
We conclude the paper in section \ref{sec:conclusions}.

\section{Technical Background}
\label{sec:background}

\subsection{Statistical Finite Elements for Linear PDEs}
\label{subsec:statfem-linear}
The \acrlong*{statFEM} constructs a prior distribution over the finite element coefficients of a \gls*{PDE} solution, based on an additive \gls*{GP} forcing error. For linear differential operators, this induced distribution is itself a Gaussian, and can be found in closed form. Consider an elliptic \gls*{PDE} with homogeneous Dirichlet boundary conditions \cite{evans2010partial}, with linear differential operator $\mathcal{L}_\theta$, parametrised by $\theta$, acting on the solution $u\coloneqq u(x)$ defined over some domain $x \in \Omega$, 
\begin{align}
    \mathcal{L}_\theta u &= f + \epsilon, \quad x\in \Omega\nonumber\\
    u &= 0, \quad x\in \partial \Omega,
\end{align}
with the forcing $f\in L^2(\Omega)$, and the additive \gls*{GP} noise by $\epsilon \sim \mathcal{GP}(0, k)$ with covariance function $k \colon \Omega \times \Omega \rightarrow \mathbb{R}$. 
To convert the strong form of the \gls*{PDE} to the weak form we multiply both sides by a test function $v \in V$, where $V$ is an appropriate function space (e.g. $H^1_0(\Omega)$, the Sobolev space of square-integrable functions with square integrable first-order weak derivatives) and integrate over the domain,
\begin{align}
    \int_{\Omega}(\mathcal{L}_\theta u) v \md x = \int_{\Omega} (f + \epsilon)v\md x, \quad \forall v\in V.
\end{align}
For the elliptic \gls*{PDE}, the order of differentiation can be reduced via integration by parts yielding a bilinear form $\mathcal{A}_\theta (u,v)$. This gives the following variational form \begin{align}
    \mathcal{A}_\theta (u,v) = \langle f + \epsilon, v \rangle, \quad \forall v\in V,
\end{align}
where $\langle \cdot, \cdot \rangle$ is the $L^2(\Omega)$ inner product. The \gls*{FEM} method forms a discrete approximation with a finite-dimensional set of basis functions defined over the domain, $\left\{\phi_{j}(x)\right\}_{j=1}^{n_u}$. We search for solutions in the subspace $V_h \subset V$ defined by the span of these basis functions $V_{h} = \mathrm{span}\left\{\phi_{j}(x)\right\}_{j=1}^{n_u}$, where the parameter $h$ refers to the degree of mesh-refinement of the domain. The basis function expansion $u_h(x) = \sum_{j=1}^{n_u}\hat{u}_j\phi_j(x)$ provides a finite-dimension variational problem
\begin{align}
    \mathcal{A}_\theta (u_h, \phi_j) = \sum_{i=1}^{n_u}\hat{u}_i\mathcal{A}_\theta (\phi_i, \phi_j) = \left\langle f, \phi_j \right\rangle + \left\langle \epsilon, \phi_j \right\rangle, \quad \forall j\in \{1,\dots, n_u\},
\end{align}
which can be rewritten as the following finite-dimensional linear system
\begin{align}\label{eq:fem-poisson}
    \mA_\theta\vu = \vb + \ve, \quad  \ve \sim \mathcal{N}\left(\mathbf{0}, \mG\right),
\end{align}
with $\vu = \left[\hat{u}_1, \hat{u}_2, \dots, \hat{u}_{n_u}\right]^{\top}$, $\left[\mA_\theta \right]_{ij} = \mathcal{A}_\theta (\phi_i, \phi_j)$, $\left[\vb\right]_i = \langle f, \phi_i\rangle$, $\left[\mG\right]_{ij} = \langle \phi_i, \langle k(\cdot, \cdot), \phi_j \rangle\rangle$. Since Equation \eqref{eq:fem-poisson} is linear and the additive noise is Gaussian, we have a closed form expression for the prior over \gls*{FEM} coefficients that is also Gaussian and defined by
\begin{align}\label{eq:statfem-prior-elliptic}
    p(\vu | \theta, \vb) = \mathcal{N}\left(\mA_\theta^{-1}\vb, \mA_\theta^{-1} \mG \mA_\theta^{-\top}\right).
\end{align}
Provided the bilinear form $\mathcal{A}_\theta$ is coercive on $V$ the problem is well-posed, the stiffness matrix $\mA_\theta$ will be positive definite, implying its inverse exists (see e.g.~\cite{ern2004theory}[Remark~2.20]). For setting the homogeneous Dirichlet boundary conditions, we set the rows corresponding to boundary nodes to the row of the identity matrix, and the forcing vector to zero. 
This can be adapted for inhomogeneous boundary conditions via boundary condition lifting (see e.g.~\cite{hughes2003finite}[Page~8]), which adapts the entries of the right-hand side vector to the non-zero boundary condition.
We choose to constrain the covariance function $k$ to be zero on the boundary, so that boundary conditions are imposed exactly. To achieve this, we form an approximation to the \gls*{GP}, as seen in \cite{solinHilbertSpaceMethods2020}, which has the benefit in our case of being constrained to zero covariance on the boundary. The eigenfunctions of the Laplacian operator (essentially the vibrational modes of the domain) form the basis, which are weighted based on the spectrum of the covariance function. This approximation is given by
\begin{align}\label{eq:hilbert-covariance-approx}
    k(x,x') \approx \sum_{l=1}^{m} S\left(\sqrt{\lambda_l}\right) g_l(x) g_l(x')
\end{align}
where the function $g_l(x) = \sum_{l=1}^{n_u}\tilde{g}_{li}\phi_i(x)$ is the $L^2$-normalised \gls*{FEM} approximation of the $l^{\text{th}}$ eigenfunction of the Laplacian operator, $\lambda_l$ the corresponding eigenvalue, and $m$ the rank of the approximation. The weighting function $S(\cdot)$ is the spectral density of the covariance function $k$. Note this requires an isotropic covariance function, i.e. such that we can write $k(x, x') = k(\Vert x - x'\Vert)$. This is detailed in Appendix \ref{appendix:hilbert-space-GP}.

\subsection{Statistical Finite Elements via Langevin Dynamics}
\label{subsec:technical-ULA}
Having derived the prior density, we can construct a Langevin \gls*{SDE} that takes samples from this distribution by defining the potential $\Phi_{\theta}(\vu, \vb) \coloneqq -\log p(\vu | \theta, \vb)$. The Langevin \gls*{SDE} with this potential as the drift has the prior distribution as its limiting distribution:
\begin{align}
    \md \vu_{t} = -\nabla_{\vu}\Phi_{\theta}(\vu_t, \vb)\md t + \sqrt{2} \, \md \mathbf{B}_t,
\end{align}
where $\left(\mathbf{B}_t\right)_{t\geq 0}$ is a $n_u$-dimensional Brownian motion. The Euler-Maruyama discretisation is used to simulate from this Langevin \gls*{SDE}, which approximates the transition density between successive time-points with a Gaussian distribution \cite{platen2010numerical}. The associated discrete-time Markov chain
\begin{align}\label{eq:langevin-markov-chain}
    \vu_{k+1} = \vu_k - \gamma\nabla_{\vu}\Phi_{\theta}(\vu_k, \vb) + \sqrt{2\gamma} \, \boldsymbol{\zeta}_{k+1}
\end{align} 
where $\left\{\boldsymbol{\zeta}_k\right\}_{k\in \mathbb{N}}$ is a sequence of i.i.d. standard Gaussian random variables, and $\gamma>0$ controls the degree of time-discretisation. This discretisation induces a bias, the degree of which is determined by the step-size $\gamma$. For some applications this bias may be unacceptable, and can be corrected by a Metropolis-Hastings style accept/reject step for each iteration of the algorithm (this is known as \gls*{MALA} \cite{girolami2011riemann}) under the condition that the unnormalised density can be evaluated directly, not just the score as for the unadjusted Langevin algorithm (ULA). \gls*{ULA} does not perform this accept/reject stage and is therefore more computationally efficient than \gls*{MALA}, at the cost of introducing discretization bias \cite{dalalyan2017theoretical, dalalyan2019user, durmus2019high}.

For the examples in this paper, the observation operator $\mathcal{H}$ is described by the point-wise evaluation of the \gls*{PDE} solution at a set of observation locations in the domain, $\left\{x_{j}\right\}_{j=1}^{n_y}$. The discrete observation operator, $\mH \colon \mathbb{R}^{n_u} \rightarrow \mathbb{R}^{n_y}$, maps the \gls*{FEM} coefficients to observations by interpolating the \gls*{FEM} solution onto these observation locations. We also assume an additive Gaussian noise on the observations, which sets up the linear discrete observation model
\begin{align}
    \vy = \mH \vu + \vr, \quad \vr \sim \mathcal{N}\left(\mathbf{0}, \mR\right),
\end{align}
where $\mR \in \mathbb{R}^{n_y\times n_y}$ is the noise covariance matrix.
In the context of Bayesian inference, the target is the \textit{posterior} distribution, $p(\vu|\vy, \theta, \vb) = p(\vu, \vy, \theta, \vb)/\int p(\vu, \vy, \theta, \vb)\md \vu$, which in general has an intractable normalisation constant due to the integral $\int p(\vu, \vy, \theta, \vb)\md \vu$. Since \gls*{ULA} only relies on the score of the target density, samples from the posterior can be generated without having to compute the normalisation constant \cite{welling2011bayesian}. 
For this posterior inference, the joint potential of latent variables and data $\Phi_{\theta}^y(\vu, \vb) \coloneqq \Phi_{\theta}(\vu, \vb) -\log p(\vy|\vu)$ is used; the invariant distribution of the following Langevin \gls*{SDE} is posterior $p(\vu | \vy, \theta, \vb)$
\begin{align}\label{eq:langevin_posterior}
    \md \vu_{t} = - \nabla_{\vu}\Phi_{\theta}^y(\vu_t, \vb) \, \md t + \sqrt{2} \, \md \mathbf{B}_t
\end{align}
where $\left(\mathbf{B}_t\right)_{t\geq 0}$ is a $n_u$-dimensional Brownian motion.

\subsection*{Preconditioned diffusion} For some problems, the potential may be ill-conditioned requiring small step-sizes for numerical stability, leading to inefficient sampling. This can be alleviated by preconditioning the Langevin dynamics with a symmetric \textit{preconditioner}, $\mP\in\mathbb{R}^{d\times d}$ \cite{girolami2011riemann}. The preconditioned Langevin \gls*{SDE} is simply,
\begin{align}\label{eq:pULA-SDE}
    \md \vu_{t} = -\mP\nabla_{\vu}\Phi_{\theta}^y(\vu_t, \vb)\md t + \sqrt{2} \, \mP^{1/2} \md \mathbf{B}_t.
\end{align}
Preconditioning can improve the stability and convergence of the sampling scheme, whilst leaving the stationary measure invariant \cite{akyildiz2022statistical}. Section \ref{subsec:linear-poisson} shows how preconditioning is applied for posterior sampling.

\subsection{Interacting Particle Langevin Dynamics} 
\label{subsec:technical-IPLA}
Whilst sections below will detail the \gls*{IPLA} algorithm for forcing estimation and solving the inverse problem, we provide a brief overview here. \gls*{IPLA} simultaneously evolves a set of particles with Langevin dynamics, which are used to estimate the gradient of the marginal likelihood with respect to parameters. A system of $N$ particles, $\{\vu_t^{(n)}\}_{n=1}^{N}$ are simulated from the Langevin dynamics described in \eqref{eq:langevin_posterior}, with independent Brownian motions $(\mathbf{B}_{t}^{(n)})_{t\geq 0}$ for $n = 1, \ldots, N$. These particles and parameters are jointly performing gradient flow on the negative joint log likelihood where the parameter process is driven with a scaled Brownian motion $\sqrt{{2}/{N}}\md \mathbf{B}_{t}^{(0)}$. This noise term, which is the main difference from the methods introduced in \cite{kuntz2023Particle}, allows for non-asymptotic analysis, and in \cite{akyildiz2025Interacting}, it is shown that given Lipschitz and strong convexity assumptions on the potential, the resulting parameter estimate will converge to the true \gls*{MMLE}, and the rate of convergence is derived in terms of the Lipschitz and strong convexity constants, and time-discretization step-size. The method has also shown to be adaptable to the nondifferentiable and underdamped cases, see, e.g. \cite{encinar2024proximal,oliva2024kinetic}. 
Sections \ref{subsec:linear-poisson} and \ref{subsec:inverse-problem} give the resulting interacting particle systems for estimation of the forcing function and solving the inverse problem, respectively.

\section{Interacting Particle Langevin Dynamics for Linear Poisson PDE}
\label{sec:ipla-linear-poisson}
One example of an elliptic \gls*{PDE} is the Poisson equation, which describes processes such as steady-state heat diffusion. The strong form of the linear Poisson equation with homogeneous diffusivity coefficient $\varkappa$, is given by  
\begin{align}\label{eq:linear-poisson-new}
    -\nabla \cdot (\varkappa\nabla u(x)) &= f + \epsilon, \quad x\in \Omega\\
    u(x) &= 0, \quad x\in\partial\Omega.
\end{align}
In our case we use the exponential mapping to reparameterise the strong form with diffusivity $\varkappa$, by $\theta = \log\varkappa$ to ensure non-negativity.
The bilinear form associated with this strong form, $\mathcal{A}_\theta(u,v) = e^\theta\left\langle \nabla u, \nabla v \right\rangle$, produces the stiffness matrix $\left[\mA_\theta\right]_{ij} = e^\theta\left\langle \nabla \phi_i, \nabla \phi_j \right\rangle$. Since the bilinear form is symmetric, the stiffness matrix is symmetric positive definite, and the linear system, $\mA_\theta \vu = \vb$, will have a unique solution. 

The joint likelihood model, $p(\vy,\vu|\vb,\theta)$, factorises as the product of the following Gaussian distributions
\begin{align}
    p(\vu | \theta, \vb) = \mathcal{N}\left(\mA_{\theta}^{-1}\vb, \mA_{\theta}^{-1} \mG \mA_{\theta}^{-\top}\right),\quad
    p(\vy | \vu) = \mathcal{N}\left(\mH \vu, \mR\right),
\end{align}
which defines the potential $\Phi_{\theta}^{\vy}(\vu,\vb)$, which can be used for posterior sampling or for \gls*{MMLE} of either $\vb$, $\theta$ using \gls*{IPLA}.

To further enable \gls*{MMAP} estimation of the forcing $\vb$, and the parameter $\theta$, we require prior distributions over these quantities. We will assume independence $p(\vb, \theta) = p(\vb)p(\theta)$, as we will be estimating these individually. The discretised prior distribution over the forcing coefficients, $\vb$, can be induced from a \gls*{GP} prior over the forcing function. 
Given $f \sim \mathcal{GP}(\mu_f, k_f)$, we have $\left[\boldsymbol{\mu}\right]_i = \langle \phi_i, \mu_f \rangle$, $\left[\boldsymbol{\Sigma}\right]_{ij} = \langle \phi_i, \langle k_f(\cdot, \cdot),\phi_j \rangle \rangle$. For the parameter $\theta$ we assume a Gaussian prior, which is equivalent to a log-normal prior over the diffusivity, but with the added constraint of non-negativity when performing gradient update steps. The prior distributions are as follows
\begin{align}
p(\vb) = \mathcal{N}\left(\boldsymbol{\mu}, \boldsymbol{\Sigma}\right), \quad
p(\theta) = \mathcal{N}\left(\mu_{\theta}, \sigma^2_{\theta}\right).
\end{align}
We have now defined each component of the full joint probabilistic model $p(\vy, \vz, \vu, \vb, \theta) = p(\vy|\vu)p(\vu|\vb,\theta)p(\vb)p(\theta)$, which we use to define the joint potential for the \gls*{MMAP} estimation, using the potential in \eqref{eq:langevin_posterior} for posterior sampling, and the prior distributions $p(\vb), p(\theta)$:
\begin{align}
    \Psi_{\theta}^y(\vu,\vb) := \Phi_{\theta}^y(\vu,\vb) - \log p(\vb) - \log p(\theta).
\end{align}
Using the forcing estimation as an example to illustrate the standard \gls*{IPLA}, we present the interacting Langevin \gls*{SDE} for estimation of $\vb$, with $\theta$ fixed between iterations. For a system of $N$ particles, $\{\vu_t^{(n)}\}_{n=1}^{N}$, the interacting Langevin \gls*{SDE} has the following form
\begin{align}
    \md\vb_{t} &= - \frac{1}{N}\sum_{n=1}^{N}\nabla_{\vb}\Psi_{\theta}^y(\vu_{t}^{(n)}, \vb_t)\md t + \sqrt{\frac{2}{N}} \, \md \mathbf{B}_t^{(0)} \label{eq:ipla-sde-b}\\
    \md\vu_{t}^{(n)} &=  - \nabla_{\vu} \Psi_{\theta}^y(\vu_t^{(n)}, \vb_t)\md t + \sqrt{2} \, \md\mathbf{B}_t^{(n)}, 
 \quad \forall n\in\{1,\dots,N\}, \label{eq:ipla-sde-u}
\end{align}

where $(\mathbf{B}_t^{(n)})_{t\geq0}$ is an independent Brownian motion for each $n$. This can be discretised to form the update steps, which are given below for the k\textsuperscript{th} iteration
\begin{align}
    \vb_{k+1} &= \vb_k - \frac{\gamma}{N}\sum_{n=1}^{N}\nabla_{\vb}\Psi_{\theta}^y(\vu_{k}^{(n)}, \vb_k)+ \sqrt{\frac{2\gamma}{N}} \, \boldsymbol{\zeta}_{k+1}^{(0)}\label{eq:ipla-forcing-b}\\
    \vu_{k+1}^{(n)} &= \vu_k^{(n)} - \gamma\nabla_{\vu} \Psi_{\theta}^y(\vu_k^{(n)}, \vb_k) + \sqrt{2\gamma} \, \boldsymbol{\zeta}_{k+1}^{(n)}, \quad \forall n\in\left\{1, \dots, N\right\},\label{eq:ipla-forcing-u}
\end{align}
where for each $n\in\{1,\dots,N\}$, $\{\boldsymbol{\zeta}_k^{(n)}\}_{k\in \mathbb{N}}$ is a sequence of i.i.d. standard Gaussian random variables.
This can simply be adapted to \gls*{MMAP} estimation of the diffusivity, by instead fixing $\vb$, and replacing \eqref{eq:ipla-sde-b} to evolve $\theta$ with the gradient of the joint potential $\nabla_{\theta}\Psi_{\theta}^y(\vu, \vb)$. In the following two sections we derive the particular forms of the \gls*{IPLA} update steps for both the forcing estimation and solving the inverse problem, including preconditioning.

\subsection{Forcing Estimation for Linear Poisson PDE}
\label{subsec:linear-poisson}
Here, we derive the \gls*{IPLA} update steps for estimating the forcing $\vb$, while taking the diffusivity as known and constant. 
We also apply a fixed preconditioner to this potential to improve stability and sample efficiency. The preconditioners for the latent sampling $\mathbf{P}_{\vu}$, and for parameter gradient descent $\mathbf{P}_{\vb}$ are set to the inverse of the block diagonal Hessian matrices (see Equation \eqref{eq:hessian-poisson-forcing}) as follows:
\begin{align}\label{eq:block-preconditioning-new}
  \mathbf{P}_{\vb} = \left[\mG^{-1} + \boldsymbol{\Sigma}^{-1}\right]^{-1},
  \quad
  \mathbf{P}_{\vu} = \left[\mA_{\theta}^{\top}\mG^{-1}\mA_{\theta} + \mH^{\top}\mR^{-1}\mH\right]^{-1},
\end{align}
which individually preconditions both the posterior sampling and the parameter gradient descent with their exact Hessians, but ignores their interaction (off-diagonal terms of the full block Hessian). With this choice of preconditioning the forcing estimation simplifies into the following update steps
\begin{align}
    \vb_{k+1} &= (1-\gamma)\vb_k + \gamma\mathbf{P}_{\vb}\left(\mG^{-1}\mA_{\theta}\left(\frac{1}{N}\sum_{n=1}^{N}\vu_{k}^{(n)}\right) 
    + \boldsymbol{\Sigma}^{-1}\boldsymbol{\mu}\right)
    + \sqrt{\frac{2\gamma}{N}} \mathbf{P}_{\vb}^{1/2} \boldsymbol{\zeta}_{k+1}^{(0)},\label{eq:forcing-preconditioning-b}\\
    \vu_{k+1}^{(n)} &= (1-\gamma)\vu_{k}^{(n)}  + \gamma \mathbf{P}_{\vu}\left( \mA_{\theta}^{\top}\mG^{-1}\vb_k + \mH^{\top}\mR^{-1}\vy\right) + \sqrt{2\gamma}\mathbf{P}_{\vu}^{1/2}\boldsymbol{\zeta}_{k+1}^{(n)},\label{eq:forcing-preconditioning-u} 
\end{align}
for all $n \in \{1, \dots, N\}$, with $\{\boldsymbol{\zeta}_k^{(n)}\}_{n=0}^N$ as i.i.d. standard $n_u$-dimensional Gaussians.

\subsection{Solving the Inverse Problem for Linear Poisson PDE}
\label{subsec:inverse-problem}
Here, we derive the update steps for estimating the diffusivity, given a known forcing function.
Preconditioning is applied at each time-step to the particle Langevin dynamics, based on the current value of $\theta_k$, where $\mP_\vu [\theta_k] =  \left[\mA_{\theta_k}^{\top}\mG^{-1}\mA_{\theta_k} + \mH^{\top}\mR^{-1}\mH\right]^{-1}$. Writing out the \gls*{IPLA} update steps for this problem, we have
\begin{align}
    \theta_{k+1} &= \theta_k - \gamma \left( n_u + \frac{\theta_k - \mu_\theta}{\sigma_\theta^2}\right) - \frac{\gamma}{N}\sum_{n=1}^{N}  \left(\mA_{\theta_k} \vu_k^{(n)} - \vb\right)^{\top} \mG^{-1}\mA_{\theta_k} \vu_k^{(n)} + \sqrt{\frac{2\gamma}{N}} \, \boldsymbol{\zeta}_{k+1}^{(0)}\label{eq:ipla-inverse-theta}\\
    \vu_{k+1}^{(n)} &= (1-\gamma)\vu_k^{(n)} + \gamma \mP_{\vu}[\theta_k] \left( \mA_{\theta_k}^{\top} \mG^{-1} \vb + \mH^\top \mR^{-1} \vy\right) + \sqrt{2\gamma} \, \mP_{\vu}[\theta_k]^{1/2} \boldsymbol{\zeta}_{k+1}^{(n)}\label{eq:ipla-inverse-u}
\end{align}
for all $n\in\left\{1, \dots, N\right\}$, with $\boldsymbol{\zeta}_k^{(0)} \sim \mathcal{N}(0,1)$, and $\{\boldsymbol{\zeta}_k^{(n)}\}_{n=1}^N$ as i.i.d. standard $n_u$-dimensional Gaussians. The factorisation $\mA_\theta = e^\theta \mA$ allows us to compute the gradient $\nabla_\theta\Psi_\theta^y$ for \eqref{eq:ipla-inverse-theta} analytically. For problems without such a factorisation, we could use automatic differentiation to compute $\frac{\partial}{\partial \theta}\mA_\theta$ for use in the gradient update steps.

\section{Convergence analysis}
\label{sec:theoretical-results}
In this section, we provide convergence analysis of some of the methods developed in this paper. The \gls*{IPLA} method is guaranteed to convergence if the joint potential is strongly convex and globally gradient-Lipschitz continuous. We show below that for the example of elliptic \gls*{PDE} forcing estimation, we can \textit{prove} that these conditions hold and, as a result, the \gls*{IPLA} algorithm will converge to the \gls*{MMAP} estimate.

\subsection{Convergence of elliptic PDE forcing estimate to MMAP estimator}
\label{subsec:convergence_forcing_MMAP}
Recall that we consider the setting in section~\ref{subsec:linear-poisson} and in this case, we have the potential $\Psi^y_{\theta}(\vu, \vb)$ for the joint system $\vv = (\vu, \vb)$ conditioned on the data $\vy$. We can write the gradients of the potential w.r.t. the solution, $\vu$, and parameters $\vb$, are as follows
\begin{align}
    \nabla_{\vu} \Psi^y_{\theta}(\vu, \vb) &= \mA_{\theta}^{\top}\mG^{-1}\left(\mA_{\theta} \vu - \vb\right) + \mH^{\top}\mR^{-1}\left(\mH \vu - \vy\right), \label{eq:poisson-grad-u}\\
    \nabla_{\vb} \Psi_{\theta}^y(\vu, \vb) &= -\mG^{-1}\left(\mA_{\theta} \vu - \vb\right) + \boldsymbol{\Sigma}^{-1}(\vb - \boldsymbol{\mu}).\label{eq:poisson-grad-b} 
\end{align}
For this problem we have a quadratic potential, and so the strong convexity and Lipschitz constants are determined by the minimum and maximum eigenvalues of the Hessian:
\begin{align}\label{eq:hessian-poisson-forcing}
\nabla^2 \Psi_{\theta}^y = \begin{bmatrix}
        \mA_{\theta}^{\top}\mG^{-1}\mA_{\theta} + \mH^{\top}\mR^{-1}\mH & -\mA_\theta^{\top}\mG^{-1}\\
        -\mG^{-1}\mA_{\theta}  & \mG^{-1} + \boldsymbol{\Sigma}^{-1}
    \end{bmatrix}.   
\end{align}
Provided the Hessian is positive-definite, we have $\mu>0$, and a unique minimiser. We can indeed prove this result for this particular problem.
\begin{lemma}\label{lem:assump_satisfied} The potential $\Psi_{\theta}^y(\vu, \vb)$ for the linear-Poisson equation is $\mu$-strongly convex in $(\vu, \vb)$, with $\mu = \lambda_{\min}(\nabla^2\Psi_{\theta}^y) > 0$.
\end{lemma}
\begin{proof} See Appendix~\ref{app:strong-convexity}.
\end{proof}
 We also show that without the prior, considering only the \gls*{MMLE} potential $\Phi_\theta^y$, we can only guarantee $\mu = \lambda_{\min}(\nabla^2\Phi_\theta^y)\geq 0$ and therefore is only convex, not \textit{strongly} convex. This motivates the use of the prior on the forcing function to ensure identifiability and guarantee a unique minimiser. 
 For this linear case, we can determine this minimiser analytically and we derive this in Appendix \ref{appendix:map-solution}. 

\begin{theorem}\label{thm:convergence_forcing} Assume $L := \lambda_{\max}(\nabla^2\Phi_\theta^y) < \infty$. Then, the parameter marginal $(\vb_k)_{k\geq 1}$ generated by the \gls*{IPLA} algorithm \eqref{eq:ipla-forcing-b}--\eqref{eq:ipla-forcing-u} with the step-size satisfying $\gamma \leq 2 / (\mu + L)$ for estimating the forcing function in the linear Poisson equation converges to the \gls*{MMAP} estimator $\vb^{\star} = \arg\max p(\vb | \theta, \vy)$. In particular, we have the following convergence result
\begin{align}\label{eq:parameter-convergence}
    \mathbb{E} \left[\left\Vert \vb_k - \vb^{\star} \right\Vert^2\right]^{1/2} \leq \sqrt{\frac{2 n_u}{N \mu}} + (1 - \mu \gamma)^{k} W_2(\nu_0, \tilde{\pi}) + 1.65 \frac{L}{\mu} \sqrt{\frac{(N+1)n_u}{N}} \gamma^{1/2},
\end{align}
where $\nu_0$ is the initial distribution of the particle system so that $W_2(\nu_0, \tilde{\pi}) < \infty$ where $\tilde{\pi}$ is given in the proof, $\mu$ is the strong convexity constant given in Lemma~\ref{lem:assump_satisfied}.
\end{theorem}
\begin{proof}
See Appendix~\ref{proof:thm:convergence_forcing}.
\end{proof}
We experimentally confirm the result in Theorem~\ref{thm:convergence_forcing}, by estimating the order of convergence of the \gls*{IPLA} estimate to the \gls*{MMAP} estimate in section \ref{subsec:forcing-linear-results}.
A few remarks to fully unpack this result are in order.
\begin{remark}\label{rem:complexity} Given Theorem~\ref{thm:convergence_forcing}, we can obtain complexity guarantees for our method to estimate the forcing. Let ${\kappa} = L/\mu$ denote the condition number. We start by requiring $N = \mathcal{O}(\varepsilon^{-2} n_u)$ which makes the first term in \eqref{eq:parameter-convergence} of order $\varepsilon$, i.e., $\mathcal{O}(\varepsilon)$. Next, we set $\gamma = \mathcal{O}(\varepsilon^2 n_u^{-1} \kappa^{-2})$ which makes the last term in \eqref{eq:parameter-convergence} of order $\mathcal{O}(\varepsilon)$. Finally, setting $k \geq \widetilde{\mathcal{O}}(n_u \kappa^2 \varepsilon^{-2})$ provides\footnote{The notation $\widetilde{\mathcal{O}}$ suppresses logarithmic factors.} an error so that $\mathbb{E} \left[\left\Vert \vb_k - \vb^{\star} \right\Vert^2\right]^{1/2} \leq \varepsilon$. \hfill $\diamond$
\end{remark}
Here, we note the linear dependence of the complexity on the latent dimension $n_u$, that is for lower latent dimension (i.e. for coarser \gls*{FEM} discretisation) the complexity will reduce, which may appear counter-intuitive. This is a result of bounding the error between the \gls*{MMAP} estimate $\vb^\star$ and the \gls*{IPLA} estimate $\vb_k$; with $n_u$ too small, the resulting \gls*{MMAP} estimate will not necessarily be close to the true forcing $\vb$ due to the large \gls*{FEM} discretisation error incurred. We point readers towards the works~\cite{karvonen2022error,papandreou2023Theoretical}, which provide error bounds for \gls*{statFEM} between the true solution and the \gls*{statFEM} distributions, in terms of the \gls*{FEM} discretisation.
\begin{remark} Remark~\ref{rem:complexity} makes it clear that the number of iterations $k$ to obtain $\varepsilon$-accuracy is strongly influenced by $\kappa$. It is thus crucial to control this number for ill-conditioned problems, as in our case. The algorithm in \eqref{eq:forcing-preconditioning-b}--\eqref{eq:forcing-preconditioning-u} is \textit{preconditioned} and this can potentially improve our bound. Specifically, with a transformation of variables $\mathbf{w}\coloneqq \mP^{-1/2} \vv$, it can be shown via It\^{o}'s formula that an equivalent \gls*{SDE}, with potential $\Psi^y_{\theta,\mP}(\mathbf{w})\coloneqq \Psi_{\theta}^y(\vv) = \Psi^y(\mP^{1/2}\mathbf{w})$ has the same invariant measure \cite{akyildiz2022statistical}. We can then analyse the Lipschitz and strong convexity constants in the transformed coordinate system, using the same method as in section~\ref{subsec:convergence_forcing_MMAP} now that the \gls*{SDE} is of the same form. The Hessian for the transformed potential $\nabla_{\mathbf{w}}^2\Psi^y_{\theta,\mP} (\mathbf{w}) = \mP^{1/2}\nabla_{\vv}^2\Psi_{\theta}^y(\vv)\mP^{1/2}$, and so the constants relating to the preconditioned system are
\begin{align}
    \mu_{\mP} = \lambda_{\mathrm{min}}\left(\mP^{1/2}\nabla^{2}\Psi^{y}_{\theta}\mP^{1/2}\right),
    \quad 
    L_{\mP} = \lambda_{\mathrm{max}}\left(\mP^{1/2}\nabla^{2}\Psi_{\theta}^{y}\mP^{1/2}\right).
\end{align}
The bound in Theorem~\ref{thm:convergence_forcing} then applies directly with the constants $\mu_{\mP}$ and $L_{\mP}$. Let $\kappa_{\mP} = L_{\mP}/\mu_{\mP}$ be the condition number of the preconditioned system. Given Remark~\ref{rem:complexity} and assuming that $\kappa_{\mP} \ll \kappa$, preconditioning has two clear effects: First, the step-size $\gamma = \mathcal{O}(\varepsilon^2 n_u^{-1} \kappa_{\mP}^{-2})$ can now be bigger (within the restrictions provided in Theorem~\ref{thm:convergence_forcing}). Secondly, to obtain $\varepsilon$-accuracy, we may require significantly fewer steps, i.e., $k = \widetilde{\mathcal{O}}(n_u \kappa_{\mP}^2 \varepsilon^{-2})$. \hfill $\diamond$
\end{remark}

\section{Experimental Results}
\label{sec:experimental-results}
In this section, we provide comprehensive numerical evidence that the methods we presented above are effective in practice. We consider two examples, with different levels of complexity, to demonstrate the effectiveness of the \gls*{IPLA} algorithm within the setting of inverse problems. Specifically, we consider the following examples:
\begin{itemize}
    \item \textbf{Linear Poisson forcing estimation:} We include this example because both the exact posterior, and the exact \gls*{MMAP} estimator are available analytically due to the linearity of the problem. This allows us to experimentally confirm the convergence rate of our \gls*{IPLA} estimate of the forcing to the ground truth \gls*{MMAP} estimator. We also assess the stability of the sampler with varying time-discretisation, and with and without preconditioning.
    \item \textbf{Linear Poisson inverse problem:} We consider the estimation of an unknown constant diffusivity for the linear Poisson problem, demonstrating the effect of a warm-start on the efficiency of estimating the parameter.
\end{itemize}
\subsection{Forcing estimation for Linear Poisson PDE}
\label{subsec:forcing-linear-results}
Here we consider two example domains, the unit interval, and a 2D disc. In both cases we fix the diffusivity to a constant value of one, i.e. $\varkappa=1$. 
For the unit interval the linear Poisson \gls*{PDE} is given by
\begin{align}\label{eq:pde1}
    -\nabla \cdot (\nabla u(x)) &= f + \epsilon, \quad x\in (0,1), \quad \epsilon \sim \mathcal{GP}(0, k)\nonumber\\
    u(x) &= 0, \quad x\in\left\{0,1\right\},
\end{align}
where $f(x) = 5\sin (6\pi x)$, $k(x,x') = \exp\left(-\frac{\Vert x - x'\Vert^2_2}{2 * 0.1^2}\right)$, and for the forcing prior $f\sim\mathcal{GP}(0,k_f)$, $k_f(x,x') = 4\exp\left(-\frac{\Vert x - x'\Vert^2_2}{2 * 0.1^2}\right)$. 
For the disc, $\Omega = \left\{\vx \in \mathbb{R}^2 \colon \Vert\vx\Vert < 1 \right\}$, we have:
\begin{align}\label{eq:pde2}
    -\nabla \cdot (\nabla u(\vx)) &= f + \epsilon, \quad \vx\in\Omega, \quad \epsilon \sim \mathcal{GP}(0, k)\nonumber\\
    u(\vx) &= 0, \quad \vx\in\partial\Omega,
\end{align}
with $f(\vx) = 100\exp\left(-\frac{\Vert \vx + \vm\Vert^2_2}{2 * 0.2^2}\right)+100\exp\left(-\frac{\Vert \vx - \vm\Vert^2_2}{2 * 0.2^2}\right)$, $\vm=[0.4,0.4]^\top$, and $k(\vx,\vx') = \exp\left(-\frac{\Vert \vx - \vx'\Vert^2_2}{2 * 0.1^2}\right)$. For inference $f\sim\mathcal{GP}(0,k_f)$, $k_f(\vx,\vx') = 100\exp\left(-\frac{\Vert \vx - \vx'\Vert^2_2}{2 * 0.2^2}\right)$. For the examples in this work we assume the \gls*{GP} hyperparameters, such as kernel length-scale and amplitude are known \textit{a priori}. For the forcing prior hyperparameters, these would be set by the practitioner to reflect the \textit{a priori} knowledge. The \gls*{GP} misspecification hyperparameters would ideally be estimated from data. In this study we do not include such estimation, which could be performed by e.g. extending the \gls*{MMLE}/\gls*{MMAP} to include the kernel hyperparameters, via grid search, or Bayesian optimisation methods; we leave this as a consideration for future work.

\subsubsection*{Convergence results}
The convergence of the estimated forcing to the \gls*{MMAP} estimator is computed by assuming a convergence of the form $\Vert \vb^{(N)}_k - \vb^\star \Vert_2 = CN^{-p}$, where $\vb^{(N)}_{k}$ is the parameter estimate after $k$ iterations of preconditioned \gls*{IPLA} \eqref{eq:forcing-preconditioning-b}--\eqref{eq:forcing-preconditioning-u} using $N$ particles.
The \textit{rate} of convergence, $C$ will depend on the problem and is not the target of this analysis while the \textit{order} of convergence, $p$ is of interest and can be numerically approximated. The logarithm of the error is given by $\log \left( \Vert \vb_k^{(N)} - \vb^\star \Vert\right) = \log C - p\log N$. 
\begin{figure}[ht!]
    \centering
    \subfloat[]{\includegraphics[width=0.5\textwidth]{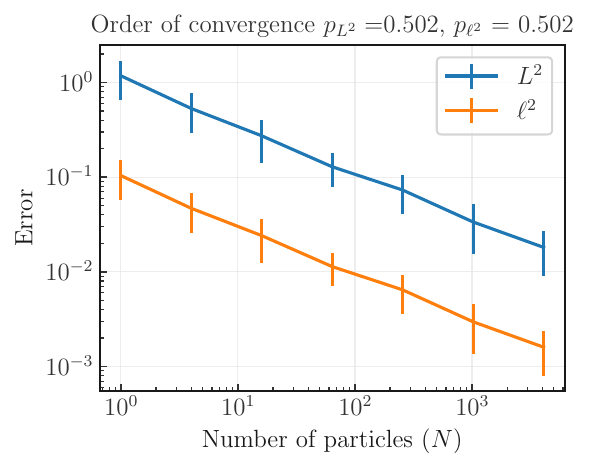}}
    \subfloat[]{\includegraphics[width=0.5\textwidth]{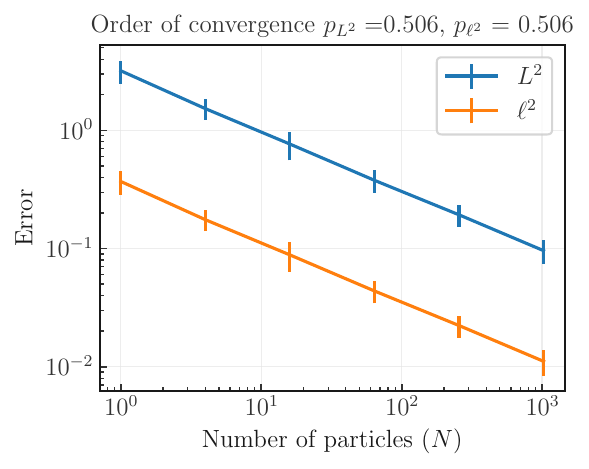}}
    \caption{Experimental validation of the order of convergence derived in Theorem \ref{thm:convergence_forcing} for geometries (a) 1D Poisson and (b) Poisson on 2D disc. Theorem \ref{thm:convergence_forcing} proves an error bound on the \gls*{MMAP} estimator, we compute this error for increasing number of particles $N$, showing $N^{-1/2}$ order of convergence. We refer to $\Vert \vb_k - \vb^\star \Vert$, $\Vert f_k - f^\star \Vert$ as the $\ell^2$-error and $L^2$-error, respectively. Error bars computed using 10 Monte Carlo simulations.}
    \label{fig:map_convergence}
\end{figure}
For example, consider the case where we double $N$ and take the log-ratio of the errors and take the difference of the logarithms: 
    $\log_2 \left( \Vert \vb_k^{(N)} - \vb^\star \Vert\right) -  \log_2\left(\Vert \vb_k^{(2N)} - \vb^\star \Vert\right) = p.$
Running the algorithm multiple times on randomised datasets gives us an average estimate of $\Vert \vb_k^{(N)} - \vb^\star \Vert$ for a particular value of $N$. We also compute the $L^2$-error between the FEM interpolated functions $\Vert f_k - f^\star \Vert$, as a comparison. The order of convergence for two examples (1D Poisson on unit interval, 2D Poisson on disc) can be seen in Figure~\ref{fig:map_convergence} and are estimated to be $(0.502, 0.506)$ are close to $1/2$, which is the theoretical order of convergence derived in Theorem \ref{thm:convergence_forcing}. We should also see the particle estimate of the posterior converge to the true posterior, $p(\vu|\vy,\theta,\vb^{\star})$; to assess this we compute the variance of the empirical estimate and compare this against the true posterior variance for increasing numbers of particles, which we visualise in Figure~\ref{fig:variance} for the 1D Poisson example \eqref{eq:pde1}.
\begin{figure}
    \centering
    \includegraphics[width=\textwidth]{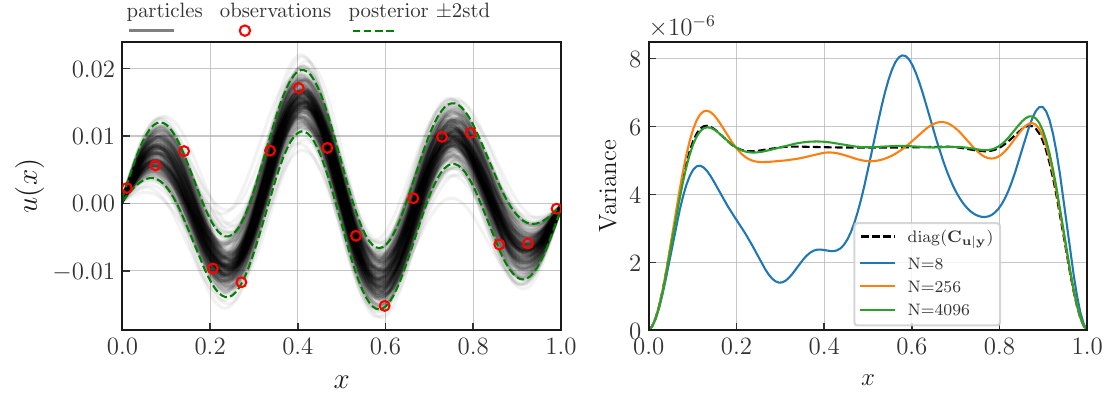}
    \caption{(\textit{left}) True posterior $p(\vu|\vy,\theta,\vb^\star) = \mathcal{N}(\vm_{\vu|\vy}, \mC_{\vu|\vy})$, compared to particle approximated posterior $p(\vu|\vy,\vb_k) = \frac{1}{N}\sum_{n=1}^{N} \delta_{\vu_{k}^{(n)}}$, for $N=256$. (\textit{right}) Variance of true posterior vs approximations for increasing number of particles.}
    \label{fig:variance}
\end{figure}
\subsubsection*{Robustness}
Ideally, the sampling should be numerically stable and resulting inference independent of the degree of mesh refinement for the problem; this corresponds to keeping the dimension of the data $n_y$ fixed, whilst increasing the latent dimension $n_u$. We have found that increasing the dimension of the latent space makes the problem more ill-posed, particularly when specifying a covariance length-scale that is larger than the \gls*{FEM} node separation making the condition number of the forcing error covariance very large. To demonstrate this, we ran the standard \gls*{ULA} algorithm, without preconditioning for the 1D linear Poisson problem \eqref{eq:pde1}, where for $\Omega = [0,1]$ we specify nodes $x_i=\frac{i}{n_u-1}, i\in\left\{0,\dots,n_u-1\right\}$, with piecewise linear basis functions, giving $\vu \in \mathbb{R}^{n_u}$. Increasing the dimension $n_u$ of the state-space and finding the maximum step-size before the sampler diverges gives an indication of the numerical stability of the sampler, results for this are shown in Table \ref{tab:mesh-refinement-ULA}.
\begin{table}[ht!]
	\centering
	\caption{Maximum step-sizes for a stable sampler for \gls*{IPLA} \textit{without} preconditioning, for increasing mesh refinement and different \gls*{statFEM} prior length-scales, $\ell$, where $k(x, x') = \sigma\exp(-\frac{\Vert x-x' \Vert^2}{ 2\ell^2})$. Fixed dataset size $n_y=16$ used, for 1D Poisson on unit interval domain.}
	\begin{tabular}{cccccc}  
		$n_u$ & 5 & 10 & 20 & 30 & 50 \\
		\midrule
		max $\gamma$ ($\ell=0.01$)    & $ 10^{-4.78} $    &  $ 10^{-5.64} $  & $ 10^{-6.87} $ & $ 10^{-7.58} $ & $ 10^{-8.61} $ \\
		\midrule
		max $\gamma$ ($\ell=0.1$)  & $ 10^{-4.78} $    &  $ 10^{-6.94} $  & $ <10^{-10} $ & $ <10^{-10} $ & $ <10^{-10} $ \\
		\bottomrule
	\end{tabular}
	\label{tab:mesh-refinement-ULA}
\end{table}
\begin{table}[ht!]
\centering
\caption{Mesh-refinement results for preconditioned \gls*{IPLA}: Condition numbers for standard and preconditioned potential for increasing $n_u$.}
\begin{tabular}{cccccc}
   $n_u$            & 32                    & 64                    & 128                   & 256                   & 512                   \\
               \midrule
$\kappa$       & $1.14 \times 10^{10}$ & $2.30 \times 10^{10}$ & $4.60 \times 10^{10}$ & $9.20 \times 10^{10}$ & $2.83 \times 10^{11}$ \\
\midrule
$\kappa_{\mP}$ & 65.6                  & 65.6                  & 65.5                  & 65.3                  & 65.0                 \\
\bottomrule
\end{tabular}
\label{tab:condition-number}
\end{table}
Preconditioning the Langevin sampling mitigates this problem. To demonstrate the robustness of preconditioning to mesh-refinement, we calculate the condition number for the standard potential, $\kappa$, and the preconditioned potential, $\kappa_{\mP}$, for increasing latent dimension. Table \ref{tab:condition-number} shows the condition number increasing significantly for the standard potential, whilst the preconditioned case the condition number is small, and constant with respect to the latent dimension.

\subsection{Estimating Diffusivity Parameter for Linear Poisson PDE}
\label{subsec:diffusivity-estimation}
For the inverse problem of estimating the diffusivity, we use the 1D Poisson equation: 
\begin{align}\label{eq:pde3}
-\nabla \cdot (e^{\theta}\nabla u) &= f + \epsilon, \quad x\in(0,1),\nonumber\\
    u(x) &= 0, \quad x\in\{0,1\}.
\end{align}
where data is generated from $f(x) = 20\sin (4\pi x)$ and with $\varkappa = 1$. For inference we use $\epsilon \sim \mathcal{GP}(0, k)$, $k(x,x') = 3\exp\left(-\frac{\Vert x - x'\Vert^2_2}{2*0.02^2}\right)$, with a prior $p(\theta)=\mathcal{N}(1.5,0.75^2)$. Estimation results are shown in Figure \ref{fig:diffusivity-estimation}, showing convergence of the parameter estimate to the true value.
\begin{figure}[ht]
    \centering
    \includegraphics[width=\textwidth]{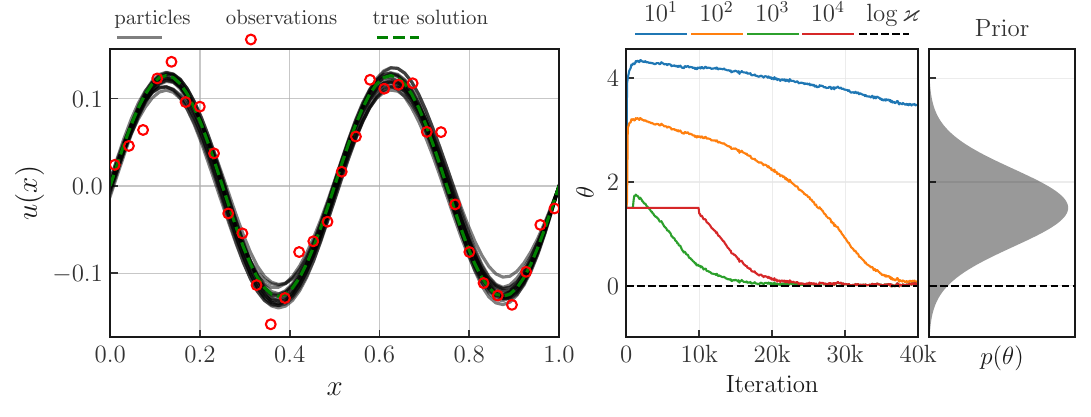}
    \caption{(\textit{left}) Particles $\{\vu_{k}^{(n)}\}_{n=1}^N$ plotted in black, alongside the solution $u(x)$ (dashed green) that generated the data $\vy$ (red circles), for warm-start length $10^3$, step-size $\gamma=0.001$. (\textit{center}) Comparison of parameter traces, for different lengths of warm-start compared to true value $\varkappa=1$, initalised at $\theta_0=1.5$. (\textit{right}) Prior distribution, $p(\theta) = \mathcal{N}(1.5, 0.75^2)$.}
    \label{fig:diffusivity-estimation}
\end{figure}
In this particular experiment, we use a warm-start strategy for $\mathbf{u}$ variables and found that without this warm-start, the algorithm might take a long time to converge. The reason for this can be seen from the gradient of the potential which reads as
\begin{align*}
\nabla_{\vu}\Psi^{y}_{\theta}(\vu, \vb) = \mA_{\theta}^{\top} \mG^{-1} \left(\mA_\theta \vu - \vb\right)+ \mH^{\top}\mR^{-1}\left(\mH \vu - \vy\right).
\end{align*}
This potential is ill-behaved in $\theta$-dimension with growing Lipschitz coefficients, creating instability. However, observing that, for an initial (fixed) point $\theta_0$, we have the strong convexity
\begin{align*}
    \left\langle \vu - \vu', \nabla_{\vu} \Psi^{y}_{\theta_0}(\vu, \vb) - \nabla_{\vu} \Psi^{y}_{\theta_0}(\vu', \vb) \right\rangle
    &\geq  \mu \Vert \vu - \vu' \Vert^2,
\end{align*}
and Lipschitz continuity
\begin{align*}
    \left\Vert \nabla_{\vu} \Psi^{y}_{\theta_0}(\vu, \vb) - \nabla_{\vu} \Psi^{y}_{\theta_0}(\vu', \vb) \right\Vert
    &\leq L \Vert \vu - \vu' \Vert, 
\end{align*}
where
\begin{align*}
    \mu = \lambda_{\mathrm{min}}\left(\mA_{\theta_0}^{\top} \mG^{-1} \mA_{\theta_0} + \mH^{\top}\mR^{-1}\mH\right)
    , \quad L = \lambda_{\mathrm{max}}\left(
    \mA_{\theta_0}^{\top} \mG^{-1} \mA_{\theta_0} + \mH^{\top}\mR^{-1}\mH \right).
\end{align*}
We can run the diffusion on $\mathbf{u}$ with fixed $\theta_0$ which will exhibit exponential convergence. Initialising $\mathbf{u}$ variables at stationarity for $\theta_0$ and running the system jointly produces a stable behaviour experimentally, which might be an interesting starting point of an analysis for this case. We refer to the iteration of the particle update \eqref{eq:ipla-inverse-u} while fixing $\theta_0$ constant a `warm-start', and the number of iterations for which this is run the `warm-start length'.
The particles are initialised at zero and the diffusivity parameter at some value larger than the truth. With shorter warm-starts, the parameter estimate increases rapidly from the initial value away from the true value because the particles have not had sufficient time to move away from zero, hence the parameter updates believe the system is highly diffusive. The second longest warm-start provides the quickest converge to the true diffusivity, with the longest spending too long in the warm-start stage. Figure \ref{fig:diffusivity-estimation} shows the parameter traces for these different lengths of warm-start, alongside the prior over the parameter $\theta$.

\section{Nonlinear Extension}
\label{sec:nonlinear}
As we have seen in section \ref{subsec:statfem-linear}, if the \gls*{PDE} is defined by a linear differential operator the prior distribution induced by the \gls*{GP} noise is also a Gaussian. This is not the case in general. A nonlinear differential operator will induce a (typically non-analytic) non-Gaussian prior over the \gls*{FEM} coefficients. In this section we derive the \gls*{statFEM} prior for a nonlinear \gls*{PDE}, before deriving the potential and approximations for a nonlinear Poisson \gls*{PDE}, before demonstrating the forcing estimation results.

\subsection{Statistical Finite Elements for Nonlinear PDEs}
\label{subsec:statfem-nonlinear}
 A stochastic \gls*{PDE} with the nonlinear differential operator $\cF$ and forcing $f$
\begin{align}\label{eq:nonlinear-PDE}
    \cF(u) = f + \epsilon, \quad
  \epsilon \sim \mathcal{GP}(0, k(\cdot, \cdot)),
\end{align}
can be discretised via the \gls*{statFEM} method giving a system of nonlinear algebraic equations $\mF(\vu) = \vb + \ve, \quad \ve \sim \mathcal{N}(\mathbf{0}, \mG)$ (see Appendix \ref{appendix:nonlinear-pde} for derivation).
The induced prior distribution can be found via a transformation of variables \cite{billingsley1995measure}
, from distribution $\mathbf{E}\sim f_{\mathbf{E}}(\ve)$ to $\mathbf{U}\sim f_{\mathbf{U}}(\vu)$ where $\vu = g(\ve)$,
\begin{align}\label{eq:transformation-of-variables}
    f_{\mathbf{U}}(\vu) = \det \left\vert \frac{\partial g^{-1}(\vu)}{\partial \vu} \right\vert f_{\mathbf{E}}(g^{-1}(\vu)).
\end{align}
Using this transformation of variables with $f_{\mathbf{E}} = \mathcal{N}(\mathbf{0},\mG)$ and $g^{-1}(\vu) = \mF(\vu) - \vb$, we can see how the resulting density of the solution depends on the nonlinear differential operator and its Jacobian:
\begin{align}\label{eq:nonlinear-statfem-prior}
    p(\vu | \vb) = \det \vert \mJ (\vu) \vert 
    \frac{1}{Z}
    \exp \left\{-\frac{1}{2}\left(\mF(\vu) - \vb\right)^{\top}\mG^{-1}\left(\mF(\vu) - \vb\right)\right\},
\end{align}
where the constant $Z = (2\pi)^{-n_u/2} \det\vert \mG \vert^{-1/2}$, and the Jacobian of the nonlinear transform $\mF$ is denoted $\mJ(\vu) = \frac{\partial \mF}{\partial \vu}\big\vert_{\vu}$. This prior, in general, is non-Gaussian due to the Jacobian determinant term. 

\subsection{Forcing Estimation for Nonlinear Poisson PDE}
\label{subsec:nonlinear-poisson}
The strong form of a nonlinear Poisson equation is shown in Equation \eqref{eq:nonlinear-poisson}, where we introduce the nonlinearity through a solution-dependent diffusivity, $q(u)$: 
\begin{align}\label{eq:nonlinear-poisson}
-\nabla \cdot (q(u)\nabla u) = f + \epsilon.
\end{align}
We proceed by multiplying with a test function and integrating over the domain, before applying integration by parts to convert this to the weak form
\begin{align}\label{eq:nonlinear-weak-main}
    \int_{\Omega}q(u)\nabla u \cdot \nabla v\md x = \int_{\Omega}\left(f + \epsilon\right)v \md x. 
\end{align}
For full details see Appendix \ref{appendix:nonlinear-pde}. With a basis function expansion $u_h = \sum_{i=1}^{n_u}\hat{u}_{i}\phi_i$, this is rewritten in vector form
\begin{align}\label{eq:nonlinear-discrete}
    \mF(\vu) = \vb + \ve, \quad \ve \sim \mathcal{N}(\mathbf{0}, \mG), 
\end{align}
where we have $\vu = \left[\hat{u}_1, \hat{u}_2, \dots, \hat{u}_{n_u}\right]^{\top}$, $\left[\mF(\vu)\right]_i = \left\langle q(u_h)\nabla u_h , \nabla \phi_i\right\rangle$, $\left[\vb\right]_i = \langle f, \phi_i\rangle$, and $\left[\mG\right]_{ij} = \langle \phi_i, \langle k(\cdot, \cdot), \phi_j \rangle\rangle$. The prior potential we define as $\Phi(\vu, \vb) \coloneqq - \log p(\vu|\vb)$, and posterior potential $\Phi^y(\vu,\vb) \coloneqq \Phi(\vu,\vb) - \log p(\vy|\vu)$ using the same observation model as \eqref{eq:observation-model}. Finally, in the same manner as the linear Poisson case we use a \gls*{GP} prior for the forcing function, and define the joint potential as $\Psi^y(\vu,\vb) \coloneqq \Phi^y(\vu,\vb) - \log p(\vb)$.

In this nonlinear case the computation of the gradient of the prior potential with respect to $\vu$, $\nabla_{\vu} \Phi(\vu, \vb)$ is intractable due to the log-determinant Jacobian term in the score of the nonlinear \gls*{statFEM} prior, so we turn to approximations of this distribution to make progress. Approximating this distribution with a Gaussian alleviates the problem of intractability, providing a linear score computed using the approximate mean and covariance. We present here the algorithm for sampling based on a Gaussian approximation of the nonlinear \gls*{statFEM} prior $p(\vu|\vb_k) \approx \mathcal{N}(\vm_k, \mC_k)$:
\begin{align}
    \vb_{k+1} &= (1-\gamma)\vb_k + \gamma\mathbf{P}_{\vb}\left( \mG^{-1}\left(\frac{1}{N}\sum_{n=1}^{N}\mF(\vu_{k}^{(n)})\right) + \boldsymbol{\Sigma}^{-1}\boldsymbol{\mu}\right) + \sqrt{\frac{2\gamma}{N}} \mathbf{P}_{\vb}^{1/2} \zeta_{k+1}^{(0)}\label{eq:nonlinear-b}\\
\vu_{k+1}^{(n)} &= (1-\gamma)\vu_{k}^{(n)}  + \gamma \mathbf{P}_{\vu}\left( \mC_k^{-1}\vm_k + \mH^{\top}\mR^{-1}\vy\right) + \sqrt{2\gamma}\mathbf{P}_{\vu}^{1/2}\zeta_{k+1}^{(n)}, \quad \forall n \in \left\{1, \dots, N\right\}\label{eq:nonlinear-u}\\
  &\text{with}\quad \mathbf{P}_{\vb} = \left[\mG^{-1} + \boldsymbol{\Sigma}^{-1}\right]^{-1},
  \quad
  \mathbf{P}_{\vu} = \left[\mC_k^{-1} + \mH^{\top}\mR^{-1}\mH\right]^{-1}.
\end{align}
Note the gradient with respect to the forcing can be computed exactly because the potential is quadratic in $\vb$, i.e. $\nabla_{\vb} \Psi^{y}(\vu, \vb) = \mG^{-1}\left(\mF(\vu) - \vb\right)$. 
Simple approaches to approximating the distribution of nonlinearly transformed Gaussian random variables \cite{hendeby2007nonlinear} include a first-order Taylor approximation, the unscented transform, and an empirical approximation to the first two moments with a sample based approach (See Table \ref{tab:nonlinear-approximations}).
The update steps~\eqref{eq:nonlinear-b}--\eqref{eq:nonlinear-u} require the computation of the preconditioner at each iteration as the approximation to $p(\vu|\vb_k)$ is updated. This is a computationally expensive step, which requires the matrix inverse $\mC_k^{-1}$ requiring order $n_u^3$ operations, followed by the inverse $\left[\mC_k^{-1} + \mH^{\top}\mR^{-1}\mH\right]^{-1}$, again order $n_u^3$. We note that for first-order Taylor, we can compute the inverse of covariance matrix efficiently from the Jacobian by noting $\mC_k^{-1} = \mJ(\vm_k)^{\top}\mG^{-1} \mJ(\vm_k)$ which is only order $n_u^2$. For larger-scale problems when this may become infeasible to recompute the preconditioner at each iteration other methods can be used for preconditioning, since for the first-order Taylor method, the preconditioning computation is order $n_u^3$ but the approximation is order $n_u^2$, hence fixing this preconditioner for some iterations may have significant impacts on efficiency. For example, a fixed preconditioned could be used, computed at the prior mean, or the preconditioner could be recomputed after a fixed number of steps. These different preconditioning strategies are of interest in future work.

 We remark that all these methods require us to solve the nonlinear system, and the first-order Taylor method requires computation of the Jacobian. This procludes the use of the first-order Taylor method for non-differentiable nonlinearities, where we can only evaluate $\mF$ and not its Jacobian $\mJ$; for this case, fixed-point methods could be used to solve $\mF(\vu)=\vb$ and the unscented transform or Monte Carlo methods applied for approximating the covariance. For this example, since we have access to the Jacobian, we apply Newton's method for solving the system.
\begin{table}[ht]
\centering
\caption{Gaussian approximations to nonlinearly transformed Gaussian random variables: first-order Taylor (FOT), unscented transform (UT) and Monte Carlo (MC). FOT requires one solve followed by computation of the Jacobian, $\mJ(\hat{\vm}) = \frac{\partial \mF}{\partial \vu}\vert_{\hat{\vm}}$. UT requires $2n_u+1$ solves, details in Appendix \ref{appendix:ut}. MC requires $M$ solves, where $\tilde{\vb}^{(j)} \sim \mathcal{N}(\mathbf{b}, \mathbf{G})$.}
\begin{tabular}{llll}
      & Solve & Mean: $\vm$            &   Covariance: $\mC$                        \\
      \toprule
FOT &
  $\mF(\hat{\vm}) = \vb$ & $\hat{\vm}$
   &
  $\mJ(\hat{\vm})^{-1}\mG \mJ(\hat{\vm})^{-\top}$ \\
 \midrule 
 UT \small{$j\in \{-n_u,\dots,n_u\}$} & $\mF(\mathbf{u}^{(j)}) = \vs^{(j)}$              & $\sum_{j}w_{j} \vu^{(j)}$ & $\sum_{j} \hat{w}_j(\vu^{(j)} - \vm)(\vu^{(j)} - \vm)^{\top}$ \\
 \midrule
MC $j\in \{1,\dots,M\}$ &
  $\mF(\mathbf{u}^{(j)}) = \tilde{\vb}^{(j)}$&
  $\frac{1}{M} \sum_{j=1}^{M}\vu^{(j)}$ &
  $\frac{1}{M-1} \sum_{j=1}^{M}(\vu^{(j)} - \vm)(\vu^{(j)} - \vm)^{\top}$\\
  \bottomrule
\end{tabular}
\label{tab:nonlinear-approximations}
\end{table}
Details of each of these approximations are given in Appendix \ref{appendix:ut}.

\subsection{Nonlinear Poisson Experimental Results}
\label{subsec:forcing-nonlinear-results}
This example introduces a nonlinearity into the Poisson diffusivity, where we can compare different nonlinear approximation methods. Here, we compare the first-order Taylor, unscented transform and Monte Carlo approximations for approximating nonlinear transforms of Gaussian random variables.
For the nonlinear Poisson \gls*{PDE} \eqref{eq:nonlinear-poisson}, we introduce a known nonlinearity with a solution-dependent diffusivity $q(x) = 1+u(x)^2$:
\begin{align}\label{eq:pde4}
    -\nabla \cdot ((1+u(x)^2)\nabla u(x)) &= f + \epsilon, \quad x\in (0,1), \quad \epsilon \sim \mathcal{GP}(0, k)\nonumber\\
    u(0) &= 0,\, u(1) = 1,
\end{align}
where $f(x) = 10\sin (2\pi x)$, $k(x,x') = \exp\left(-\frac{\Vert x - x'\Vert^2_2}{2*0.02^2}\right)$ for model misspecification, and $f \sim \mathcal{GP}(0, k_f)$, with  $k_f(x,x') = 6\exp\left(-\frac{\Vert x - x'\Vert^2_2}{2*0.1^2}\right)$ for the forcing prior.
\begin{figure}[ht!]
    \centering
    \includegraphics[width=\textwidth]{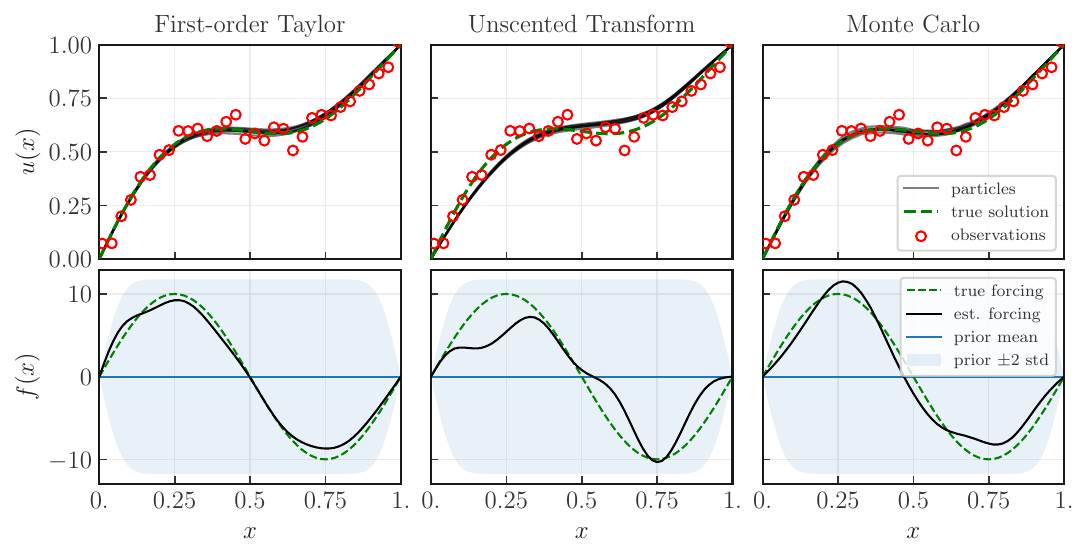}
    \caption{Comparison of forcing estimation results for different Gaussian approximations (See Table \ref{tab:nonlinear-approximations}) to the nonlinear \gls*{statFEM} prior \eqref{eq:nonlinear-poisson}. First-order Taylor approximation is the fastest to compute, and provides the best estimate of the true forcing function. Results shown at $k=10^4$ iterations, with $\gamma=0.005$ and $N=4$.}
    \label{fig:nonlinear-approximation-estimation}
\end{figure}
\begin{figure}[ht!]
    \centering
    \includegraphics[width=\textwidth]{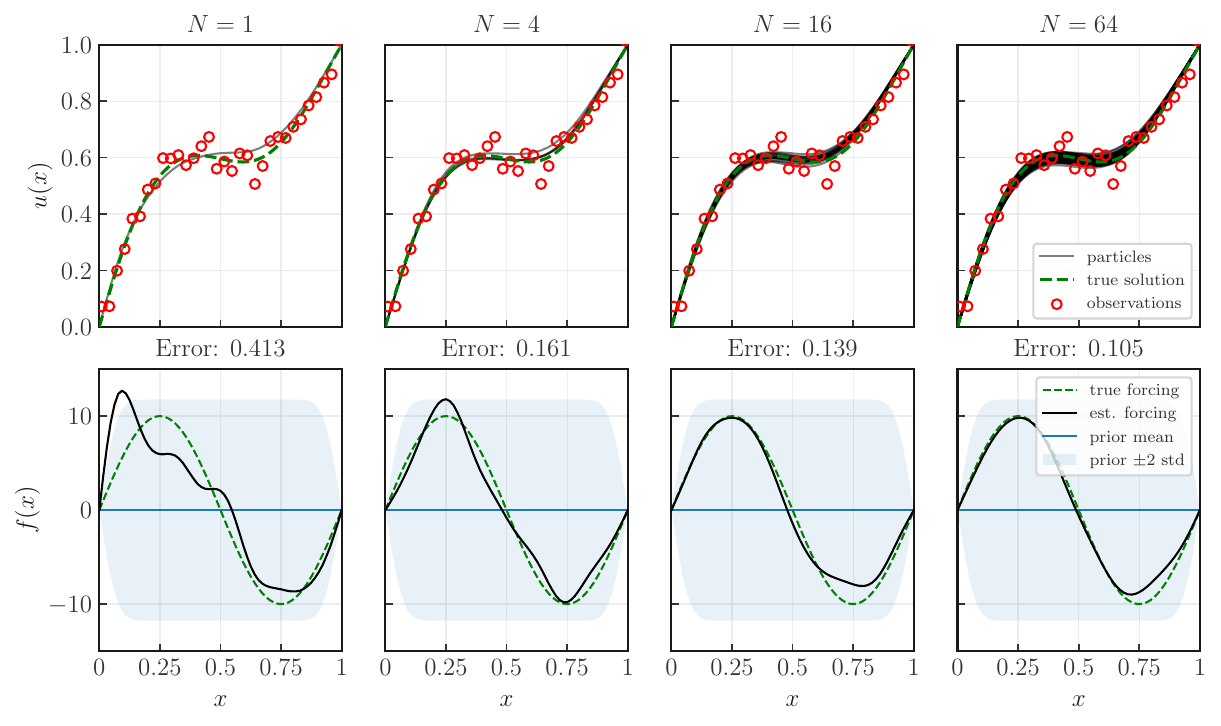}
    \caption{Estimation results for the first-order Taylor approximation to nonlinear \gls*{statFEM} prior for increasing number of particles $N$. (\textit{above}) Particles $\{\vu_{k}^{(n)}\}_{n=1}^N$ plotted in black, alongside the solution $u(x)$ (dashed green) that generated the data $\vy$ (red circles). (\textit{below}) We can see steady improvement of the forcing estimate $f_k(x) = \sum_{i=1}^{n_u}[\vb_k]_i \phi_i(x)$, to the true forcing function $f(x)=10\sin(2\pi x)$. Results shown at $k=10^4$ iterations, with $\gamma=0.005$.}
    \label{fig:nonlinear-forcing-convergence}
\end{figure}
We compare the three nonlinear approximations in Table \ref{tab:nonlinear-approximations}, showing parameter estimation and posterior inference results in Figure \ref{fig:nonlinear-approximation-estimation}. The first-order Taylor approximation has the smallest $L^2$-error between the estimated and true forcing, and is considerably less computationally expensive. Using this approximation method we assess the convergence of the forcing estimate to the true forcing with increasing particles (Figure \ref{fig:nonlinear-forcing-convergence}).

\section{Conclusions}
\label{sec:conclusions}
In this paper we have posed the problem of estimating \gls*{PDE} parameters as a \gls*{MMAP} problem, constructing \gls*{statFEM} generative model, to which we applied \gls*{IPLA}. 
Nonasymptotic bounds on the \gls*{MMAP} error have been proved, and experimentally confirmed for both 1D and 2D geometries. The approach in \cite{akyildiz2022statistical} has been extended to a {nonlinear} \gls*{statFEM} setting, and applied to the nonlinear Poisson \gls*{PDE}. Different methods for constructing Gaussian approximations to the non-Gaussian \gls*{statFEM} prior were considered and tested, with a particular focus on first-order Taylor approach. The effects of a 'warm-start' were explored for the case of solving the inverse problem of estimating diffusivity, showing that a warm-start can improve the convergence of the estimate. With nonasymptotic bounds derived, further research will consider scaling up this approach to work with larger systems, where solving the forward problem is prohibitively expensive --- in this case, Langevin dynamics can still be used to sample from the target density without forward solves \cite{akyildiz2022statistical}. Future work can also benefit from extensions to non-log-concave cases arising in nonlinear \glspl*{PDE} using the techniques introduced in \cite{zhang2023nonasymptotic,akyildiz2024nonasymptotic}.
Recent developments in particle-based schemes for accelerating the gradient descent are of interest in future work, such as the underdamped Langevin extension to \gls*{IPLA}~\cite{oliva2024kinetic}, the momentum-based approach of~\cite{lim2023momentum}, and the Stein variational gradient descent method of~\cite{sharrock2024tuning}.

\section*{Acknowledgements}
AGD was supported by Splunk Inc. [G106483] Ph.D scholarship funding. CD and MG were supported by EPSRC grant EP/T000414/1. MG was supported by a Royal Academy of Engineering Research Chair, and Engineering and Physical Sciences Research Council (EPSRC) grants EP/T000414/1, EP/W005816/1, EP/V056441/1, EP/V056522/1, EP/R018413/2, EP/R034710/1, and EP/R004889/1.

\noindent The data that support the findings of this study are available from the corresponding author, AGD, upon reasonable request.

\bibliography{references}

\pagebreak
\appendix
\section{Proof of Lemma~\ref{lem:assump_satisfied}}
\label{app:strong-convexity}
For a symmetric block matrix $M$ of the form 
\begin{align*}
    M = 
    \begin{bmatrix}
    A & B\\
    B^\top & C
    \end{bmatrix}
\end{align*}
we can use the Schur complement to determine the characteristics of this matrix in terms of its blocks, namely the fact it is positive definite (PD). We use the following from \cite{boyd2004convex}[Section~A.5.5]:
\begin{align*}
    \text{(1)}& \quad M\succ 0 \iff C \succ 0 \text{ and } A-BC^{-1}B^\top \succ 0,\\
    \text{(2)}& \quad \text{If }  C\succ 0, \text{ then, } M \succeq 0 \iff A-BC^{-1}B^\top \succeq 0.
\end{align*}
Now for the Hessian matrix we can assess if it is positive definite or positive semi definite (PSD) using these conditions, which will determine is we have $\mu>0$ for strong convexity. 
Define
\begin{align*}
\nabla^2\Psi_\theta^y = 
    \begin{bmatrix}
        \mA_{\theta}^{\top}\mG^{-1}\mA_{\theta} + \mH^{\top}\mR^{-1}\mH & -\mA_\theta^{\top}\mG^{-1}\\
        -\mG^{-1}\mA_{\theta}  & \mG^{-1} + \boldsymbol{\Sigma}^{-1}
    \end{bmatrix}
\end{align*}
Now since $\mG, \mSigma$ are covariance matrices, they are symmetric positive definite (SPD), i.e. $\mG\succ 0, \mSigma\succ 0$, as are their inverses $\mG^{-1}\succ 0, \mSigma^{-1}\succ 0$, which implies $\mG^{-1} + \mSigma^{-1} \succ 0$.  This meets the condition of $C\succ 0$ in (1), (2). Now we can check if the Schur complement is PD or PSD which will imply PD or PSD for the Hessian $\nabla^2\Psi_\theta^y$. Denote by $\mathbf{S}_\Psi$ the Schur complement of $\nabla^2\Psi_\theta^y$:
\begin{align*}
    \mathbf{S}_\Psi &= \mA_{\theta}^{\top}\mG^{-1}\mA_{\theta} + \mH^{\top}\mR^{-1}\mH -\mA_\theta^{\top}\mG^{-1}[\mG^{-1} + \boldsymbol{\Sigma}^{-1}]^{-1}\mG^{-1}\mA_{\theta}\\
    & = \mA_{\theta}^{\top}\left[ \mG^{-1} - \mG^{-1}[\mG^{-1} + \boldsymbol{\Sigma}^{-1}]^{-1}\mG^{-1}\right]\mA_{\theta} + \mH^{\top}\mR^{-1}\mH\\
    & = \mA_{\theta}^{\top}\left[ \mG + \mSigma\right]^{-1}\mA_{\theta} + \mH^{\top}\mR^{-1}\mH,
\end{align*}
where we have applied Woodbury's matrix inversion lemma 
\begin{align*}
    \mG^{-1} - \mG^{-1}[\mG^{-1} + \boldsymbol{\Sigma}^{-1}]^{-1}\mG^{-1} = \left[\mG + \mSigma\right]^{-1} ,
\end{align*}
for the second equality. 
Now we show that $\mH^\top\mR^{-1}\mH \succeq 0$ since $\mR^{-1}$ is SPD, which means it admits a Cholesky decomposition, $\mR^{-1} = \mathbf{L}_\mR\mathbf{L}_\mR^\top$, giving
\begin{align*}
    \vx^\top \mH^\top \mR^{-1}\mH\vx = \vx^\top \mH^\top \mathbf{L}_\mR\mathbf{L}_\mR^\top \mH\vx = \Vert \mathbf{L}_\mR^\top \mH \vx \Vert_2^2 \geq 0, \quad \forall \vx\in \bR^{n_u}, \vx \neq 0.
\end{align*}
We can also show that $\mA_{\theta}^{\top}[\mG + \mSigma]^{-1}\mA_{\theta}$ is PD, as $\left[\mG + \mSigma\right]^{-1}$ is SPD and $\mA_\theta$ is non-singular. We define $\vx = \mA_\theta \vy$
\begin{align*}
    \vx^\top \mA_{\theta}^{\top}\left[\mG + \mSigma\right]^{-1}\mA_{\theta}\vx = \vy^\top \mathbf{L} \mathbf{L}^\top\vy = \Vert \mathbf{L}^\top  \vy \Vert_2^2 > 0, \quad \forall \vy\in \bR^{n_u}, \vy \neq 0,
\end{align*}
which holds because we have for every $\vx\in\bR^{n_u}$, we have $\vy = \mA_{\theta}^{-1} \vx$. Putting this together, we can show the Schur complement of the Hessian is SPD:
\begin{align*}
    \vx^\top\mathbf{S}_\Psi\vx &= \vx^\top \mA_{\theta}^{\top}\left[\mG + \mSigma\right]^{-1}\mA_{\theta}\vx + \vx^\top \mH^\top \mR^{-1}\mH\vx\\
    &\geq \vx^\top \mA_{\theta}^{\top}\left[\mG + \mSigma\right]^{-1}\mA_{\theta}\vx\\
    &>0,
\end{align*}
which implies $\nabla^2\Psi_\theta^y\succ 0$ from (1), hence $\exists\mu>0$ where $\mu = \lambda_{\text{min}}(\nabla^2\Psi_\theta^y)$. 
It is also worth noting that without the inclusion of the prior on $\vb$, when just considering the \gls*{MMLE} potential $\Phi_\theta^y$, we have 
\begin{align*}
\nabla^2\Phi_\theta^y = 
    \begin{bmatrix}
        \mA_{\theta}^{\top}\mG^{-1}\mA_{\theta} + \mH^{\top}\mR^{-1}\mH & -\mA_\theta^{\top}\mG^{-1}\\
        -\mG^{-1}\mA_{\theta}  & \mG^{-1}
    \end{bmatrix}
    ,
\end{align*}
and the Schur complement becomes 
\begin{align*}
    \mathbf{S}_\Phi &= \mA_{\theta}^{\top}\mG^{-1}\mA_{\theta} + \mH^{\top}\mR^{-1}\mH -\mA_\theta^{\top}\mG^{-1}\mG\mG^{-1}\mA_{\theta}\\
    & = \mH^{\top}\mR^{-1}\mH \succeq 0,
\end{align*}
which we have shown is only PSD, which implies via (2) that $\nabla^2\Phi_\theta^y\succeq 0$, and hence there may not exist a positive minimum eigenvalue of the Hessian, which only implies convexity, and not \textit{strong} convexity.

\section{Proof of Theorem~\ref{thm:convergence_forcing}}\label{proof:thm:convergence_forcing}
Our proof is based on Theorem~1 of \cite{dalalyan2019user}. It involves re-writing \gls*{IPLA} as an instance of \gls*{ULA} and then invoking standard theoretical guarantees. Let us define a target measure on $\mathbb{R}^{(N+1) n_u}$
\begin{align}\label{eq:target_IPLA}
\tilde{\pi}(\vb, \vp_1, \ldots, \vp_N) \propto \exp\left(-\sum_{i=1}^N \Psi_{\theta}^y(\sqrt{N} \vp_i, \vb)\right).
\end{align}
It is clear that $\vb$-marginal of this measure will concentrate on the \gls*{MMAP} solution as $\tilde{\pi}_{\vb}(\vb) \propto \exp(N \log p(\vb | \vy))$ by computing $\tilde{\pi}_{\vb}(\vb) = \int \tilde{\pi}(\vb, \vp_1, \ldots, \vp_N) \md \vp_1 \ldots \md \vp_N$. It is also clear that the potential $\Psi^y_{\theta,N} (\vb, \vp_1, \ldots, \vp_N) = \sum_{i=1}^N \Psi^y_{\theta}(\sqrt{N} \vp_i, \vb)$ is $\tilde{\mu} = N\mu$ strongly convex and $\tilde{L}=NL$ gradient Lipschitz based on Lemma~\ref{lem:assump_satisfied} and \cite{oliva2024kinetic}[Lemma~A.4 and A.5]. 
The factor of $N$ comes as a result of the potential being a sum of $N$ potentials with Lipschitz and strong convexity constants $L$, $\mu$, and hence $N\mu \preceq \nabla^2 \Psi_{\theta,N}^{y} \preceq NL$.
Now, let us write the standard \gls*{ULA} for the target \eqref{eq:target_IPLA}
\begin{align*}
\md \vb_t &= - \sum_{n=1}^N \nabla_\vb \Psi_{\theta}^y(\sqrt{N} \vp_{t}^{(n)}, \vb_t) \md t + \sqrt{2} \md \mathbf{B}_t^0, \\
\md \vp_{t}^{(n)} &= - \sqrt{N} \nabla_{1} \Psi_{\theta}^y(\sqrt{N} \vp_t^{(n)}, \vb_t) \md t + \sqrt{2} \md \mathbf{B}_t^n, \quad n = 1, \ldots, N,
\end{align*}
where $\nabla_1$ denotes the gradient w.r.t.~the first argument of $\Psi^y_{\theta}$ and $(\mathbf{B}_t^n)_{t\geq 0}$ for $n = 0, \ldots, N$ are $n_u$-dimensional Brownian motions. 
Consider the Euler-Maruyama discretisation of this system with a step-size $\tilde{\gamma}$
\begin{align*}
\vb_{k+1} &= \vb_k - \tilde{\gamma} \sum_{n=1}^N \nabla_\vb \Psi_{\theta}^y(\sqrt{N} \vp_{k}^{(n)},\vb_k) + \sqrt{2\tilde{\gamma}} \mathbf{w}_{k+1}^0, \\
\vp_{k+1}^{(n)} &= \vp_k^{(n)} - \tilde{\gamma}\sqrt{N} \nabla_{1} \Psi_{\theta}^y(\sqrt{N} \vp_k^{(n)}, \vb_k) + \sqrt{2\tilde{\gamma}} \mathbf{w}_{k+1}^n, \quad n = 1, \ldots, N,
\end{align*}
where $\mathbf{w}_k^n \sim \mathcal{N}(0, I_{n_u})$ for every $k\geq 0$ and $n = 0, \ldots, N$. 
It is clear that, if we write $\vu_k^{(n)} = \sqrt{N} \vp_k^{(n)}$, and set $\tilde{\gamma} = \gamma /N$ then the above system is equivalent to \gls*{IPLA} given in eqs.~\eqref{eq:ipla-forcing-b}--\eqref{eq:ipla-forcing-u}. 
We now apply Theorem~1 of \cite{dalalyan2019user} directly with the condition $\tilde{\gamma} \leq 2/(\tilde{\mu} + \tilde{L})$, results in the bound
\begin{align*}
W_2(\nu_k, \tilde{\pi}) \leq (1 - \tilde{\mu} \tilde{\gamma})^{k} W_2(\nu_0, \tilde{\pi}) + 1.65 \frac{\tilde{L}}{\tilde{\mu}} \sqrt{(N+1)n_u} \tilde{\gamma}^{1/2}.
\end{align*}
At this point, one must note the restriction can be rewritten as follows
\begin{align*}
\tilde{\gamma} \leq 2 / (\tilde{\mu} + \tilde{L}) \iff \gamma \leq 2 / (\mu + L).
\end{align*}
Hence in practice it suffices to assume $\gamma \leq 2 / (\mu + L)$. Similarly, the bound above can be rewritten by substituting $\tilde{L} = NL$, $\tilde{\mu} = N \mu$, and $\tilde{\gamma} = \gamma / N$,
\begin{align*}
W_2(\nu_k, \tilde{\pi}) \leq (1 - \mu \gamma)^{k} W_2(\nu_0, \tilde{\pi}) + 1.65 \frac{L}{\mu} \sqrt{\frac{(N+1)n_u}{N}} \gamma^{1/2}.
\end{align*}
Merging this with Proposition~3 in \cite{akyildiz2025Interacting} which ensures that
\begin{align*}
W_2(\delta_{\vb^\star}, \tilde{\pi}_{\vb}) \leq \sqrt{\frac{2 n_u}{\mu N}},
\end{align*}
and the fact that $W_2(\nu_{k,\vb}, \tilde{\pi}_\vb) \leq W_2(\nu_k, \tilde{\pi})$, we obtain our bound by using
\begin{align*}
\bE[\|\vb_k - \vb^\star\|^2]^{1/2} &= W_2(\nu_{k,\vb}, \delta_{\vb^\star}) \leq W_2(\nu_{k,\vb}, \tilde{\pi}_\vb) + W_2(\tilde{\pi}_\vb, \delta_{\vb^\star}) \\
&\leq W_2(\nu_k, \tilde{\pi}) + W_2(\tilde{\pi}_\vb, \delta_{\vb^\star}),
\end{align*}
from which the bound follows.

\section{Forcing estimation for Linear PDE - analytic solution}
\label{appendix:map-solution}
To derive the \gls*{MMAP} forcing estimate, we start by marginalising out the latent variables $\vu$ from the joint distribution $p(\vy | \theta, \vb) = \int p(\vy|\vu) p(\vu|\theta,\vb)\md\vu$.
With models $\mA_\theta\vu = \vb + \mathbf{e}$ and $\vy = \mH \vu + \vr$, where $\mathbf{e} \sim \mathcal{N}(\mathbf{0},\mG)$ and $\mathbf{r} \sim \mathcal{N}(\mathbf{0},\mR)$ we can marginalise analytically giving $p(\vy |\theta, \vb) = \mathcal{N}\left(\mH\mA_\theta^{-1}\vb, \mH\mA_\theta^{-1} \mG \mA_\theta^{-\top}\mH^{\top} + \mR\right).$
With the conjugate prior $p(\vb) = \mathcal{N}(\boldsymbol{\mu},\boldsymbol{\Sigma})$, we can calculate the posterior $p(\vb|\vy,\theta)\propto p(\vy|\theta,\vb)p(\vb)$, where 
\begin{align}
    \log p(\vb|\vy) = -\frac{1}{2}\left(\vy - \mH\mA_\theta^{-1}\vb\right)^{\top} \left[\mH\mA_\theta^{-1} \mG \mA_\theta^{-\top}\mH^{\top} + \mR\right]^{-1} \left(\vy - \mH\mA_\theta^{-1}\vb\right) 
    \nonumber\\ \quad
    -\frac{1}{2} \left(\vb - \boldsymbol{\mu}\right)^{\top} \boldsymbol{\Sigma}^{-1} \left(\vb - \boldsymbol{\mu}\right) + \text{const},
\end{align}
where \gls*{MMAP} estimator, $\vb^{\star} = \argmax_{\vb}p(\vb|\theta,\vy)$ follows as the mean of the Gaussian posterior $p(\vb|\theta,\vy)$:
\begin{align}\label{eq:mmap-analytic}
\begin{split}
    \vb^\star = \bigg[\mA_\theta^{-\top}\mH^{\top}\left[\mH\mA_\theta^{-1} \mG \mA_\theta^{-\top}\mH^{\top} + \mR \right]^{-1} \mH\mA_\theta^{-1} + \boldsymbol{\Sigma}^{-1}\bigg]^{-1} &\\
    \bigg[\mA_\theta^{-\top}\mH^{\top}\left[\mH\mA_\theta^{-1} \mG \mA_\theta^{-\top}\mH^{\top} + \mR\right]^{-1} \vy + &\boldsymbol{\Sigma}^{-1}\boldsymbol{\mu}\bigg].
\end{split}
\end{align}
For calculating the true posterior variance in Figure~\ref{fig:diffusivity-estimation}, with $p(\vu|\vy, \vb^\star) \sim \mathcal{N}(\vm_{\vu|\vy}, \mC_{\vu|\vy})$:
\begin{align}
     \mC_{\vu|\vy} = \left[\mA_\theta^\top \mG^{-1} \mA_\theta + \mH^\top \mR^{-1} \mH\right]^{-1}, \quad \vm_{\vu|\vy} = \mC_{\vu|\vy}(\mH^\top \mR^{-1} \vy + \mA_\theta^\top \mG^{-1}\vb^\star).
\end{align}

\section{Nonlinear Poisson discretisation}
\label{appendix:nonlinear-pde}
The strong form of the nonlinear Poisson equation with Dirichlet boundary conditions reads as 
\begin{align*}
\begin{cases}
-\nabla \cdot (q(u)\nabla u) = f + \epsilon\quad &\text{in } \Omega\\
\hspace{22.3mm} u = u_D \quad &\text{on } \partial\Omega.
\end{cases}
\end{align*}
We wish to find solutions in the Sobolev space, respecting the Dirichlet boundary conditions, $V = \left\{ v \in H_0^1(\Omega)\colon v = u_D \text{ on } \partial\Omega\right\}$. Converting to the weak form proceeds as before, by multiplication with a test function $v\in \hat{V}$, where $\hat{V} = \left\{ v \in H_0^1(\Omega)\colon v = 0 \text{ on } \partial\Omega\right\}$ and integrating over the domain
\begin{align}\label{eq:nonlinear-weak-basic}
-\int_{\Omega}\nabla \cdot (q(u)\nabla u)v\md x = \int_{\Omega}\left(f + \epsilon\right)v \md x.
\end{align}
The product rule for divergence can be used to rewrite the left hand side integrand as follows
\begin{align}\label{eq:div-product-rule}
    \nabla \cdot (q(u)\nabla u)v = \nabla \cdot (q(u)\nabla u v) - q(u)\nabla u \cdot \nabla v.
\end{align}
Substituting \eqref{eq:div-product-rule} into \eqref{eq:nonlinear-weak-basic} gives
\begin{align}\label{eq:nonlinear-weak-product}
-\int_{\Omega}\nabla \cdot (q(u)\nabla u v) \md x + \int_{\Omega}q(u)\nabla u \cdot \nabla v\md x = \int_{\Omega}\left(f + \epsilon\right)v \md x.
\end{align}
The first term in \eqref{eq:nonlinear-weak-product} can be expressed as a surface integral via the divergence theorem \cite{brenner2008mathematical}, and since $v=0$ on $\partial\Omega$, this term vanishes
\begin{align}
  \int_{\Omega}\nabla \cdot (q(u)\nabla u v) \md x = \int_{\partial\Omega} (q(u)\nabla u v)\cdot \nu ds = 0,
\end{align}
where $\nu$ denotes the outward normal vector to $\partial\Omega$. The nonlinear weak form 
\begin{align}\label{eq:nonlinear-weak-final}
    \int_{\Omega}q(u)\nabla u \cdot \nabla v\md x = \int_{\Omega}\left(f + \epsilon\right)v \md x. 
\end{align}
Again we specify a basis function expansion of the solution $u_h = \sum_{j=1}^{n_u}\hat{u}_j \phi_j$, and search for solutions in the space $V_h = \mathrm{span}\left\{\phi_j\right\}_{j=1}^{n_u}$, and test with the basis functions individually giving a system of nonlinear algebraic equations in the \gls*{FEM} coefficients,
\begin{align}
    \int_{\Omega}q(u_h)\nabla u_h \cdot \nabla \phi_{j}\md x = \int_{\Omega}\left(f + \epsilon\right)\phi_j \md x, \quad \forall j\in \{1, \dots, n_u\},
\end{align}
which can be written in vector form
\begin{align}\label{eq:app-nonlinear-discrete}
    \mF(\vu) = \vb + \ve, \quad \ve \sim \mathcal{N}(\mathbf{0}, \mG),
\end{align}
where we have $\vu = \left[\hat{u}_1, \hat{u}_2, \dots, \hat{u}_{n_u}\right]^{\top}$, $\left[\mF(\vu)\right]_i = \left\langle q(u_h)\nabla u_h , \nabla \phi_i\right\rangle$, $\left[\vb\right]_i = \langle f, \phi_i\rangle$, and  $\left[\mG\right]_{ij} = \langle \phi_i, \langle k(\cdot, \cdot), \phi_j \rangle\rangle$. Section \ref{subsec:nonlinear-poisson} describes the approximations made to the \gls*{statFEM} prior, all of which require the system \eqref{eq:app-nonlinear-discrete} to be solved. This discrete problem is efficiently solved using Newton's method. 

\section{Approximations to Nonlinearly Transformed Gaussian Random Variables}\label{appendix:nonlinear-approximations}
Here we present the three methods for deriving a Gaussian approximation to a nonlinearly transformed Gaussian random variable. The transform we consider is of the form $\mF(\vu) = \vb + \ve$, where we wish to approximate $\vu$ with a Gaussian, given $\ve\sim\mathcal{N}(\mathbf{0}, \mG)$.
\subsection{First-order Taylor}\label{appendix:fot}
We linearise the system about the mean point $\vm$, found as the solution of $\mF(\vm) = \vb$. By replacing $\vu$ in \eqref{eq:app-nonlinear-discrete} with $\vm$, as follows
\begin{align}\label{eq:app-taylor-stochastic}
    \mF(\vu) - \vb\approx \mJ(\vm)(\vu - \vm), \quad \text{giving} \quad  \vu = \vm + \mJ(\vm)^{-1} \ve, \quad \ve \sim \mathcal{N}(\mathbf{0}, \mG).
\end{align}
Computing the second moment,  \begin{align}
    \mC = \bE\left[(\vu-\vm)(\vu-\vm)^{\top}\right] = \mJ(\vm)^{-1}\mathbb{E}\left[\ve\ve^{\top}\right]\mJ(\vm)^{-\top} = \mJ(\vm)^{-1}\mG\mJ(\vm)^{-\top},
\end{align}
we arrive at the approximate Gaussian distribution $\vu\sim\mathcal{N}(\vm, \mC)$.

\subsection{Unscented transform}\label{appendix:ut}
Selecting the sigma-points to evaluate the nonlinear transform requires the singular value decomposition of the covariance $\mG = \sum_{j=1}^{n_u}\sigma_j^2 \mathbf{v}_j \mathbf{v}_j^\top$, and parameters $\alpha, \beta$
\begin{align}
    \vs^{(0)} &= \vb, \quad \vs^{(\pm j)} = \vb \pm \sqrt{n_u + \lambda}\sigma_j \mathbf{v}_j\\
w^{(0)} &= \frac{\lambda}{n_u + \lambda}, \quad w^{(\pm j)} = \frac{1}{2(n_u + \lambda)}
\end{align}
where $j=1,\dots, n_u$ and $\lambda = \alpha^2(n_u + \kappa) - n_u$. The parameter $\alpha$ controls the spread of the sigma points and is typically chosen as $10^{-3}$, and $\beta$ chosen as $\beta=2$ (See~\cite{hendeby2007nonlinear}). 
\begin{align}
     \vm = \sum_{j=-n_u}^{n_u}w^{(j)}\vu^{(j)},  \quad \mC = \sum_{j=-n_u}^{n_u}\hat{w}^{(j)}(\vu^{(j)} - \vm)(\vu^{(j)} - \vm)^{\top}, 
\end{align}
where $\hat{w}^{(0)} = w^{(0)} + 1 - \alpha^2 + \beta $ and $ \hat{w}^{(j)} = w^{(j)}$ for $j\neq 0$.

\subsection{Monte Carlo}\label{appendix:mc} The first two moments of the non-Gaussian \gls*{statFEM} prior are approximated using Monte Carlo samples. The approximation is of the form $p(\vu | \vb) \approx \mathcal{N}(\mathbf{m}, \mC)$, samples are generated via $M$ forward solves of the \gls*{PDE} where the noise is sampled from its Gaussian distribution, $\mF(\mathbf{u}^{(j)}) = \mathbf{b} + \mathbf{e}^{(j)}$ with, $\mathbf{e}^{(j)} \sim \mathcal{N}(\mathbf{0}, \mathbf{G}), \forall j\in\{1, \dots, M\}$ giving the \gls*{MC} samples $\left\{\vu^{(j)}\right\}_{j=1}^M$. The statistics can then be computed empirically
\begin{align}
     \vm = \frac{1}{M} \sum_{j=1}^{M}\vu^{(j)},  \quad \mC = \frac{1}{M-1} \sum_{j=1}^{M}(\vu^{(j)} - \vm)(\vu^{(j)} - \vm)^{\top}.
\end{align}

\section{Hilbert space GP}
\label{appendix:hilbert-space-GP}

In this section we show how the Hilbert space \gls*{GP} covariance function is approximated with an \gls*{FEM} basis. To determine the functions $g_l(x)$ that form the covariance approximation in \eqref{eq:hilbert-covariance-approx}, we must solve the eigenfunction problem for the Laplacian \cite{jonesConstrainingGaussianProcesses2023, solinHilbertSpaceMethods2020}, with homogeneous Dirichlet boundary conditions
\begin{align}
\forall l\in \{1,\dots, n\}, \quad \text{  solve  }
\begin{cases}
    -\nabla^2 g_l(x) = \lambda_l g_l(x), &\quad x\in\Omega\\
    g_l(x) = 0, &\quad x\in\partial\Omega.
\end{cases}
\end{align}
We wish to solve for eigenpairs $\left\{ \lambda_l, g_l \right\}_{l=1}^{n}$.
For a generic mesh defined over a domain $\Omega$, we can solve this problem using FEM. With our \gls*{FEM} basis defined as $\left\{\phi_i\right\}_{i=1}^{n}$, we can take the inner product with the $j^{\text{th}}$ basis function, and application of integration by parts gives us the weak form
\begin{align}
    \int_{\Omega} \nabla g_l(x) \cdot \nabla \phi_j(x) \md x = \int_{\Omega} \lambda_l g_l(x) \phi_j(x) \md x.
\end{align}
If we assume a basis function approximation of the eigenfunction, i.e. $\hat{g}_l(x) = \sum_{i=1}^{n} \hat{g}_{li}\phi_i(x)$, we can write this as
\begin{align}
    \sum_{i=1}^{n} \hat{g}_{li} \int_{\Omega} \nabla \phi_i(x) \cdot \nabla \phi_j(x) \md x = \lambda_l  \sum_{i=1}^{n} \hat{g}_{li}\int_{\Omega}  \phi_i(x) \phi_j(x) \md x,
\end{align}
which can be written in the form of a generalised eigenproblem
\begin{align}\label{eq:generalised-eigenvalue}
    \mathbf{A}\hat{\mathbf{g}}_l = \lambda_l \mathbf{M} \hat{\mathbf{g}}_l,
\end{align}
where $\left[\mathbf{A}\right]_{ij} = \int_{\Omega} \nabla \phi_i \cdot \nabla \phi_j \md x$ is the stiffness matrix, $\left[\mathbf{M}\right]_{ij} = \int_{\Omega} \phi_i \phi_j \md x$ the mass matrix, and $\hat{\mathbf{g}}_l = \left[\hat{g}_{l1}, \dots, \hat{g}_{ln}\right]^{\top}$ are the eigenvectors to be solved for. Equation \eqref{eq:generalised-eigenvalue} can be solved using standard linear algebra libraries; here we use the function \texttt{scipy.linalg.eig}$(\mA, \mM)$ which provides unit $l_2$-norm eigenvectors $\hat{\mathbf{g}}_l$. The eigenfunction approximations are converted to unit $L^2$-norm functions by normalising the coefficients, i.e. $g_l(x) = \sum_{i=1}^n \tilde{g}_{li}\phi_i(x)$ where, $\tilde{g}_{li}=\hat{g}_{li}/\sqrt{\langle\hat{g}_l, \hat{g}_l \rangle}$.
In this paper we exclusively use the squared exponential kernel, parameterised by an amplitude $\sigma^2$ and length-scale $\ell$, which has the following form and associated spectral density \cite{rasmussenGaussianProcessesMachine2008}
\begin{align}
    k_{\text{SE}}(r) = \sigma^2 \exp\left(-\frac{r^2}{2\ell^2}\right), \quad
    S_{\text{SE}}(\omega) = \sigma^2 (2\pi\ell^2)^{D/2} \exp\left(-\frac{\omega^2\ell^2}{2}\right),
\end{align}
where $r = \Vert x-x' \Vert$, $x\in\mathbb{R}^D$. 
Figure \ref{fig:disc-eigenfunctions} visualises the eigenfunctions computed via \eqref{eq:generalised-eigenvalue} for the disc geometry. 
\begin{figure}[ht!]
    \includegraphics[width=\textwidth]{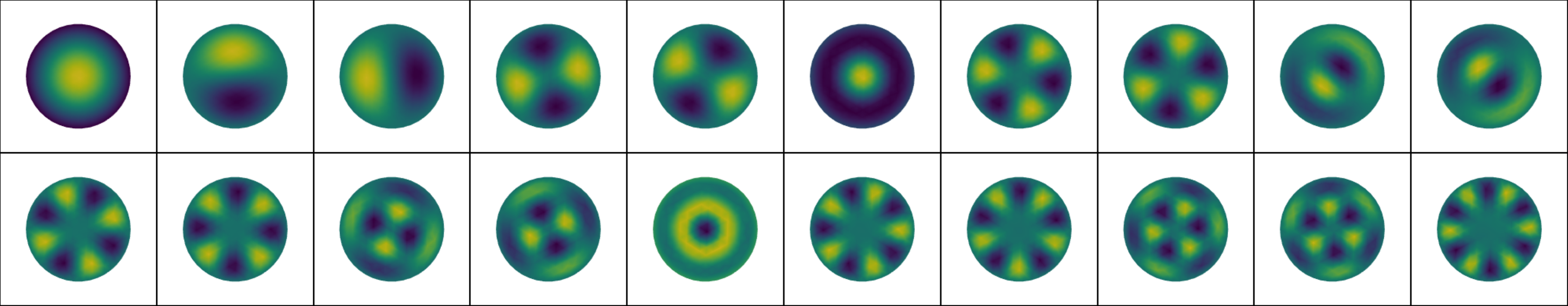}
    \caption{FEM approximations to the eigenfunctions of the Laplacian for the domain $\Omega = \left\{\vx \in \mathbb{R}^2 \colon \Vert\vx\Vert < 1 \right\}$, for the smallest 20 eigenvalues.}
    \label{fig:disc-eigenfunctions}
\end{figure}

\end{document}